\documentclass[aps,prx,reprint,superscriptaddress,nofootinbib,longbibliography,floatfix]{revtex4-2}

\usepackage[T1]{fontenc}
\usepackage[utf8]{inputenc}
\usepackage{amsmath,amssymb,amsthm}
\usepackage{bm}
\usepackage{graphicx}
\graphicspath{{}{manuscript/}}
\usepackage{framed}
\usepackage[colorlinks=true,linkcolor=blue,citecolor=blue,urlcolor=blue]{hyperref}

\newtheorem{theorem}{Theorem}
\newtheorem{definition}{Definition}
\newtheorem{proposition}{Proposition}
\newtheorem{corollary}{Corollary}

\newcommand{\Tr}{\operatorname{Tr}}

\newcommand{\Gsum}{G_{\mathrm{sum}}}

\begin{document}

\title{Operational Tube-Sector Theory of Quantum State Distinguishability Under Generalized Symmetries}

\author{Song He}
\email{hesong@nbu.edu.cn}
\affiliation{
 Institute of Fundamental Physics and Quantum Technology, and School of
 Physical Science and Technology, Ningbo University, Ningbo,
 Zhejiang 315211, China
}
%\affiliation{
% Max Planck Institute for Gravitational Physics
% (Albert Einstein Institute), Am M\"uhlenberg 1, 14476 Potsdam, Germany
%}

\date{\today}

\begin{abstract}
A variational principle for quantum-state distinguishability is established in
many-body systems with generalized symmetries, including noninvertible cases
described by fusion categories. Standard fidelity and symmetry-resolved
diagnostics emerge as coarse-grained limits of a more refined operational
structure. When symmetry actions terminate at entanglement cuts,
distinguishability is governed by boundary tube algebras within a
symmetry-constrained measurement resource theory. The physically admissible
instruments are characterized by complete positivity, entanglement-cut
locality, boundary-module covariance, and sequential stability. The resulting
optimal measurement structure is uniquely fixed by the center of the boundary
tube algebra,
\(\mathcal{A}_{\mathrm{phys}} =
Z\!\left(\mathrm{Tube}_{\mathcal{C}}(\mathcal{M}_A)\right)\),
whose primitive idempotents define tube-sector probabilities that refine
fidelity-based and symmetry-resolved descriptions. The associated tube
positive-operator-valued measures (POVMs) are extremal and yield optimal
one-shot hypothesis-testing distinguishability under symmetry constraints.
The categorical selection rule is universal across fusion categories and
independent of microscopic realization; microscopic models enter by choosing
the boundary module and the protocol family used to witness a nontrivial tube
fiber.
\end{abstract}

\maketitle

\section{Introduction}

Quantum-state distinguishability is a central concept in quantum information
and many-body physics.  It is commonly quantified by fidelity or related
overlap measures, which compress the comparison of two states to a scalar
number~\cite{Uhlmann:1976,Jozsa:1994,NielsenChuang:2000}.  This scalar
description is useful, but it is not a complete description of operational
distinguishability: state-discrimination theory and the Fuchs--van de Graaf
bounds relate fidelity to, rather than replace, measurement-dependent
distinguishability~\cite{Helstrom:1976,FuchsGraaf:1999}.  In many-body
settings one can therefore ask not only how much two reduced states overlap,
but where that overlap is located in physically accessible sectors.

A standard refinement arises in the presence of ordinary symmetries.  When
subsystem data decompose into charge or representation sectors, entanglement
and related quantities can be resolved sector by sector.  This is the
organizing idea behind symmetry-resolved entanglement
~\cite{Laflorencie:2014,Goldstein:2017bua,Xavier:2018kqb}, with extensions to
more general charges, non-Abelian symmetries, and higher-symmetry settings
~\cite{Capizzi:2022,Bianchi:2024nonabelian,Benini:2025asym,Benini:2026higherform}, as well as
categorical and noninvertible variants in conformal field theory (CFT)
~\cite{SauraBastida:2024qoc,Heymann:2024xxx,Bhattacharyya:2025cat}.  Such a
resolution is already finer than scalar overlap diagnostics, but it is still based on
labels that behave like ordinary symmetry sectors.

The missing ingredient is operational, not merely algebraic.  A generalized
symmetry can end at an entanglement cut, and the endpoint degrees of freedom
can be invisible to any scalar or ordinary charge-resolved diagnostic while
remaining accessible to a cut-local measurement.  Thus the central question is
which classical records are physically admissible at the cut, and whether this
admissible record is uniquely fixed by symmetry, locality, and positivity.
Answering this question produces a measurement principle rather than another
post hoc diagnostic: the refined outcomes are selected before choosing a
model-specific observable.

Generalized global symmetries enlarge this picture.  They are represented by
topological operators~\cite{Gaiotto:2014kfa,Luo:2023ive}; for noninvertible
symmetries these operators form fusion categories rather than groups
~\cite{Shao:2023gho,SchaferNameki:2023jdn,Bhardwaj:2023wzd} and can act locally
through quantum operations~\cite{Okada:2024}.  A key physical
feature is that defect lines may end on boundaries or entanglement cuts.  In
boundary CFT, rational conformal field theory (RCFT) defect theory, and
lattice defect constructions, such
endpoints carry boundary-module data that is not captured by ordinary group
labels~\cite{Cardy:1989ir,Fuchs:2002cm,Frohlich:2004ef,Ostrik:2003,KitaevKong:2012,Aasen:2016dop}.  Recent
boundary-tube analyses make this point sharper: the entangling cut has its own
allowed endpoint structure, and this structure can refine the sector labels
available to a measurement~\cite{Choi:2024rzz,Choi:2024ymp}.

In this work we study how this boundary-induced structure affects
distinguishability of reduced quantum states.  The basic object is the positive
overlap operator
\begin{equation}
  O_A(\rho,\sigma)=\sqrt{\rho_A}\,\sigma_A\,\sqrt{\rho_A},
  \label{eq:positive-overlap-prx}
\end{equation}
where \(\rho_A\) and \(\sigma_A\) are reduced density matrices on a subsystem
\(A\).  We do not assume that \(O_A\) itself is a directly measured observable.
The operational quantities are the positive functionals
\(\Tr(P_\alpha O_A)\), estimated by replica, swap, or tensor-network protocols
for a chosen sector readout.  Rather than introducing a new entropy as the
primary object, we ask which such sector readouts are physically meaningful at
the entanglement cut.

Our answer is a constrained hypothesis-testing statement.  The full boundary
tube algebra is an established categorical object; we do not claim it as a new
algebraic structure.  The physical question is which readout maximizes
one-shot discrimination of reduced positive-overlap statistics subject to
positivity, complete positivity of the implementing instrument, locality at
the entanglement cut, boundary-module covariance, and stability under
sequential instruments.
These requirements define an admissible readout cone, and the solution of this
variational problem is the commutative algebra associated with the center of
the boundary tube algebra.  Its primitive outcomes give a tube-resolved overlap
distribution
\[
  p_\alpha=\frac{\Tr(P_\alpha O_A)}{\Tr O_A},
\]
where \(P_\alpha\) is the projector associated with tube-sector label
\(\alpha\).

The paper turns this observation into three precise statements.  First, within
the admissible sector-readout cone defined below, for a
fixed finite semisimple boundary module, the physically admissible classical
readout of the positive-overlap statistics is the maximal commutative algebra
\(\mathcal A_{\rm phys}=Z({\rm Tube}_{\mathcal C}(\mathcal M_A))\), selected by
positivity, locality, boundary-module covariance, and stability under
sequential composition.  Equivalently, it is the variational optimum of the
restricted discrimination functional
\(\mathcal F_{\mathcal P}(\mathcal A)\) over all admissible sector readouts,
with uniqueness for tube-separating protocol families.  Second, scalar overlap and conventional
symmetry-resolved diagnostics are successive pushforwards of the corresponding
tube-sector distribution; they cannot reconstruct conditional probability
inside a forgotten tube fiber.  When the boundary module has a nontrivial
noninvertible tube fiber, this readout hierarchy is strict.  Third, if two
protocols agree after such a pushforward but differ in the conditional tube
distribution, then the difference is operationally visible as a sector-readout
hypothesis-testing advantage; any purported noncentral refinement fails at
least one admissibility constraint.  We demonstrate this mechanism in a
doubled-Ising model with a product Kramers--Wannier (product-KW) defect, where
boundary endpoints split a symmetry-trivial sector into two product-KW tube
sectors.  In this example, an unread endpoint
record changes the tube distribution without changing the scalar or
\(G=\mathbb Z_2\times\mathbb Z_2\) resolved data.  The operational claim is therefore
not that tube sectors are new superselection sectors, but that the tube center
has a distinguished operational role as the extremal measurement algebra of
reduced positive-overlap statistics.

\paragraph*{Master notation.}
The letters \(\rho,\sigma\in\mathcal D(\mathcal H)\) denote the two global
states being compared, while
\(\rho_A=\Tr_{\bar A}\rho\) and \(\sigma_A=\Tr_{\bar A}\sigma\) denote their
reductions to subsystem \(A\).  The positive-overlap operator is
\(O_A=\sqrt{\rho_A}\sigma_A\sqrt{\rho_A}\); throughout it is a positive
functional input, not a directly measured observable.  The admissible quantum
instruments are denoted by \(\mathfrak C_{\rm sym}\); they are cut-local,
quantum instruments whose outcome maps are completely positive and whose summed
channel is trace preserving, covariant under the boundary-module symmetry
action.  The induced commutative sector-readout algebras form the admissible
readout cone \(\mathfrak M_{\rm adm}\).  Its
maximal physical element is
\(\mathcal A_{\rm phys}=Z({\rm Tube}_{\mathcal C}(\mathcal M_A))\), the center
of the boundary tube algebra associated with the generalized symmetry category
\(\mathcal C\) and the entangling-cut boundary module \(\mathcal M_A\).
Tube-sector labels are \(\alpha\); coarse symmetry labels are \(q\).  The set
\(\mathcal T=\{\alpha\}\) is the label set of the primitive central
idempotents of the boundary tube algebra, i.e. the finest
symmetry-compatible tube readout for the chosen entangling-cut boundary
module.  The set \(\mathcal Q=\{q\}\) is the label set of a chosen coarser
readout, for example ordinary group-character labels.  The map
\(\pi:\mathcal T\to\mathcal Q\) is the forgetful coarse-graining map from tube
labels to coarse labels, and \(\pi_*\) is its stochastic pushforward on
probability distributions.  The
projector onto tube sector \(\alpha\in\mathcal T\) is
\(P_\alpha=\mathsf R_A(e_\alpha)\), where \(e_\alpha\) is a primitive central
idempotent and \(\mathsf R_A\) is the subsystem representation of boundary-tube
algebra elements and boundary morphisms on the Hilbert spaces associated with
\(A\).  The tube distribution is
\(p_\alpha=\Tr(P_\alpha O_A)/\Tr O_A\), the coarse distribution is
\(p_q=\sum_{\alpha\in\pi^{-1}(q)}p_\alpha\), and
\(p_{\alpha|q}=p_\alpha/p_q\) for \(p_q>0\).  For a readout algebra
\(\mathcal A\), \(D_{\rm TV}^{\mathcal A}\) denotes the total-variation (TV)
distance between its classical output distributions, and
\(\mathcal F_{\mathcal P}(\mathcal A)\) denotes the corresponding optimized
one-shot discrimination functional over a protocol family \(\mathcal P\).
Moment indices are even integers \(m=2k\), while the R\'enyi order is denoted
by \(n\).  For a finite group \(G\), \(\widehat G\) denotes its character
labels.  In the Ising example, \(1,\psi,\sigma\) denote Ising topological
sectors or primaries, and this use of \(\sigma\) is distinct from the density
matrix \(\sigma_A\).
Detailed conventions, higher moments, R\'enyi
extensions, CFT checks, and data provenance are given in the appendices.

\section{Positive Overlap and Canonical Tube Readout}

The tube sectors are defined by the entangling-cut boundary module, not by an
arbitrary basis of defect operators.  Let \(\mathcal C\) be the generalized
symmetry category and let \(\mathcal M_A\) be the boundary module selected by
the cut.  The action of \(\mathcal C\) on \(\mathcal M_A\) defines a boundary
tube algebra \({\rm Tube}_{\mathcal C}(\mathcal M_A)\)
~\cite{Etingof:2005,Ostrik:2003,KitaevKong:2012}.  The center of this
semisimple algebra is the maximal commutative algebra of sector labels that is
stable under the boundary-module action.  Its primitive central idempotents
\(e_\alpha\) are therefore the canonical symmetry-compatible outcomes,
following the standard Ocneanu tube-algebra construction of sector labels and
its tensor-categorical quantum-double/center formulation
~\cite{Ocneanu:1988,EvansKawahigashi:1998,Mueger:2003,Kitaev:2006,Bultinck:2017mnp}.
Let \(\mathsf R_A\) denote the representation of this tube algebra on the
subsystem Hilbert space.  Representing the idempotents gives the projectors
\(P_\alpha=\mathsf R_A(e_\alpha)\).

This formulation fixes what is, and is not, unique.  The primitive central
idempotents are unique up to equivalence, including Morita-equivalent
presentations of the same physical boundary module; a microscopic change of
tube basis only conjugates their Hilbert-space representation and leaves the
traces below invariant.  A different entangling-cut boundary condition can
change \(\mathcal M_A\), and hence can change the tube resolution.  That
dependence is physical boundary data, not a gauge choice.
If the symmetry is enlarged from \(\mathcal C\) to a category
\(\mathcal C'\), restriction from
\({\rm Tube}_{\mathcal C'}(\mathcal M_A)\) to
\({\rm Tube}_{\mathcal C}(\mathcal M_A)\) gives a pushforward of sector
labels; the enlarged tube distribution refines the old one rather than
contradicting it.  For an ordinary finite group the construction reduces to
the familiar charge or character projectors.  For a noninvertible symmetry,
the idempotents are boundary-tube idempotents rather than naive Fourier
transforms of defect insertions.  Thus the tube distribution is a categorical
invariant of the chosen cut and symmetry, while a group-only distribution is
one of its possible quotients.
This is the precise sense in which the categorical input is standard but the
observable question is new: we do not introduce new tube sectors, but ask how
the positive overlap probability measure changes under the quotient from tube
labels to the smaller measurement algebra available to scalar or ordinary
symmetry-resolved probes.

For a finite commutative sector-readout algebra \(\mathcal A\) with primitive
projectors \(\{E_x\}\), define the associated classical readout channel on the
positive-overlap statistics by
\begin{equation}
  \mathcal M_{\mathcal A}(O_A)_x
  =
  \frac{\Tr(E_x O_A)}{\Tr O_A}.
  \label{eq:readout-channel}
\end{equation}
For two protocols \(\lambda,\lambda'\), the one-shot success probability under
this restricted readout is
\begin{equation}
  P_{\rm succ}^{\mathcal A}(\lambda,\lambda')
  =
  \frac12\!\left[
  1+
  D_{\rm TV}\!\left(
  \mathcal M_{\mathcal A}(O_A^\lambda),
  \mathcal M_{\mathcal A}(O_A^{\lambda'})
  \right)\right].
  \label{eq:restricted-success}
\end{equation}
For a specified family \(\mathcal P\) of admissible preparation or endpoint
protocols, we define the restricted discrimination functional
\begin{equation}
  \mathcal F_{\mathcal P}(\mathcal A)
  =
  \sup_{\lambda,\lambda'\in\mathcal P}
  P_{\rm succ}^{\mathcal A}(\lambda,\lambda') .
  \label{eq:restricted-functional}
\end{equation}
This is the variational object optimized below.  The supremum formulation
separates the physical measurement principle from the particular pair of
states used in the Ising benchmark.  We call \(\mathcal P\) tube-separating if,
for every nontrivial tube fiber, it contains a pair of protocols whose
pushforward statistics agree but whose conditional tube distributions differ.
The functional is monotone under readout refinement and nonincreasing under
classical stochastic postprocessing.  The admissible cone is also closed under
classical randomization of instruments; retaining the randomization record gives
a direct-sum readout, while forgetting it is again a postprocessing.  Thus the
optimization reduces to extremal classical sector records rather than arbitrary
Hilbert-space bases.
The optimization is not over all positive-operator-valued measures (POVMs) on
\(\mathcal H_A\), which would erase the symmetry restriction, but over the admissible cone
\(\mathfrak M_{\rm adm}\) of positive readouts implemented by quantum
instruments whose outcome maps are completely positive and whose summed channel
is trace preserving, local at the entangling cut, covariant under the
boundary-module action, and stable under sequential instruments and classical
postprocessing.

We use the term ``resource theory'' in the standard operational sense: a
restricted set of operations is declared free, and access beyond that set is a
resource only insofar as it improves an operational task, here one-shot state
discrimination~\cite{Coecke:2014rsc,Chitambar:2018eyp}.  Resource theories of
measurements and measurement incompatibility provide the general
positive-operator-valued-measure (POVM)-level
language for this viewpoint~\cite{Oszmaniec:2019qme,Guff:2019rtm,Buscemi:2019qip,Guhne:2021ime},
while superselection and reference-frame restrictions motivate the use of
symmetry-constrained free operations in quantum information~\cite{Bartlett:2006rf}.
In contrast to a general measurement-resource theory, the free class below is
not chosen abstractly: it is fixed by boundary locality and
symmetry-module covariance at the entanglement cut.

\begin{definition}[Tube measurement resource theory]
\label{def:tube-resource-theory}
The admissible instruments form a resource theory.  The allowed elements of
\(\mathfrak C_{\rm sym}\) are quantum instruments whose outcome maps are
completely positive and whose summed channel is trace preserving; their
classical records are cut-local and covariant under the boundary-module action.
Free operations are classical randomization, relabeling, and
stochastic postprocessing of sector records, including the forgetful maps
\(\alpha\mapsto q=\pi(\alpha)\).  For a fixed coarse map
\(\pi:\mathcal T\to\mathcal Q\), the free pairs are those whose tube
distributions differ only through the coarse pushforward; the resource is the
conditional tube information
\begin{equation}
  \mathcal R_\pi(p,p')
  =
  D_{\rm TV}(p,p')
  -
  D_{\rm TV}(\pi_*p,\pi_*p')
  \ge 0 .
  \label{eq:tube-resource-monotone}
\end{equation}
\(\mathcal R_\pi\) is nonnegative by data processing and measures the loss of
binary total-variation distinguishability caused by the coarse map.  It
vanishes for a given pair when \(\pi\) preserves the total-variation distance
for that pair; when the coarse distributions themselves coincide, this reduces
to equality of the tube distributions on each fiber.
Within this restricted readout theory, the free operations form a closed convex
monoid under composition, classical randomization, and stochastic
postprocessing.
\end{definition}

The tube-sector POVM is therefore not introduced as a new algebraic label, but
as the maximal readout compatible with the physical restrictions defining
\(\mathfrak C_{\rm sym}\) and \(\mathfrak M_{\rm adm}\).

\begin{theorem}[Fully variational closure principle]
\label{prop:canonical-tube-readout}
Fix \(\mathcal M_A\) and assume the boundary tube algebra is finite
semisimple.  Consider the cone \(\mathfrak M_{\rm adm}\) of classical sector
readouts of reduced positive-overlap statistics whose boundary-local symmetry
content is represented by
\({\rm Tube}_{\mathcal C}(\mathcal M_A)\), and that are positive, implemented by
quantum instruments whose outcome maps are completely positive and whose summed
channel is trace preserving, local at the entangling cut, covariant under the
boundary-module action, and stable under sequential composition of instruments.
In the classical-postprocessing
preorder, the unique maximal commutative measurement algebra satisfying these
constraints is
\begin{equation}
  \mathcal A_{\rm phys}
  \simeq
  Z\!\left({\rm Tube}_{\mathcal C}(\mathcal M_A)\right).
  \label{eq:aphys-tube-center}
\end{equation}
Equivalently, every admissible sector readout is a coarse graining of the
primitive central idempotents of
\({\rm Tube}_{\mathcal C}(\mathcal M_A)\).  Hence, for every pair of
protocols,
\begin{equation}
  P_{\rm succ}^{\mathcal A}(\lambda,\lambda')
  \le
  P_{\rm succ}^{\mathcal A_{\rm phys}}(\lambda,\lambda')
  \qquad
  \forall\,\mathcal A\in\mathfrak M_{\rm adm}.
  \label{eq:pointwise-readout-optimality}
\end{equation}
Consequently,
\begin{equation}
  \mathcal A_{\rm phys}
  \in
  \arg\max_{\mathcal A\in\mathfrak M_{\rm adm}}
  \mathcal F_{\mathcal P}(\mathcal A)
  \label{eq:variational-argmax}
\end{equation}
for any protocol family \(\mathcal P\).  The optimization is closed in the
following sense.  Every admissible instrument factors through
\(\mathcal A_{\rm phys}\) followed by a free postprocessing map:
\begin{equation}
  \begin{gathered}
  \forall\,\mathcal A\in\mathfrak M_{\rm adm},\quad
  \exists\,\mathsf R_{\mathcal A}\in{\rm FreeOps},\\
  \mathcal M_{\mathcal A}
  =
  \mathsf R_{\mathcal A}\circ
  \mathcal M_{\mathcal A_{\rm phys}} .
  \end{gathered}
  \label{eq:factor-through-aphys}
\end{equation}
Here \({\rm FreeOps}\) denotes the classical randomization, relabeling, and
stochastic postprocessing maps of
Definition~\ref{def:tube-resource-theory}.  If
\(\mathcal A_1\prec\mathcal A_2\preceq\mathcal A_{\rm phys}\) in the
postprocessing preorder and \(\mathcal P\) contains a pair separated by
\(\mathcal A_2\) but not by \(\mathcal A_1\), then
\[
  \mathcal F_{\mathcal P}(\mathcal A_1)
  <
  \mathcal F_{\mathcal P}(\mathcal A_2).
\]
In particular, for tube-separating \(\mathcal P\), the maximizer is
order-theoretically unique up to null outcomes, relabeling, and
Morita-equivalent presentations of the same boundary module.
\end{theorem}

\emph{Proof sketch.}
The argument is a physical constraint chain rather than a choice of tube
basis.  First, \(O_A=\sqrt{\rho_A}\sigma_A\sqrt{\rho_A}\) is the positive
operator whose sector functionals are estimated; it is not assumed to be
directly measured as an observable.  A sector outcome must therefore define a
POVM-type positive weight, \(\Tr(P O_A)\ge0\), so admissible sector labels are
represented by positive projectors on the subsystem Hilbert space.
Second, the readout is required to be local at the entangling cut.  This
restricts generalized-symmetry operations to line segments, defect endpoints,
and boundary-changing junctions supported on the cut, namely to the boundary
tube algebra of \((\mathcal C,\mathcal M_A)\).  Third, a physical sector
measurement cannot depend on a particular presentation of the boundary module.
It must be covariant under boundary-module intertwiners, which forces the
measured labels to commute with all allowed module actions.  Thus the
classical labels lie in the center of the boundary tube algebra.  Fourth,
sequential sector measurements must again be admissible: their projectors must
be closed under composition, refinement, and coarse graining.  Since the center
is finite semisimple and commutative, this closed classical readout algebra is
generated by its primitive central idempotents.  Any further refinement would
either use nonlocal subsystem operators, break boundary-module covariance, or
fail closure as a classical sector record.  Any coarser readout is obtained by
grouping primitive idempotents.  This proves maximality and uniqueness in the
classical-postprocessing preorder, up to the standard equivalence of boundary
modules.  Since total variation distance and the associated binary
hypothesis-testing success probability cannot increase under stochastic
postprocessing, the tube-center readout pointwise dominates every admissible
coarse graining, Eq.~\eqref{eq:pointwise-readout-optimality}.  Taking the
supremum over \(\lambda,\lambda'\in\mathcal P\) gives the variational
characterization in Eq.~\eqref{eq:variational-argmax}.  This is also the
resource-theoretic completeness statement: an admissible instrument cannot
produce sector records outside the tube center without leaving
\(\mathfrak M_{\rm adm}\); within the cone it can only apply a free
postprocessing to the primitive central-idempotent record.  Tube-separating
protocol families make the optimum unique because every proper coarse graining
forgets at least one conditional tube fiber that is tested by some pair in
\(\mathcal P\), giving strict monotonicity in that preorder.

The basic observable is the normalized tube-resolved overlap distribution
\begin{equation}
  p_\alpha
  =
  \frac{\Tr(P_\alpha O_A)}{\Tr O_A},
  \qquad
  \sum_\alpha p_\alpha=1 .
  \label{eq:tube-overlap-distribution}
\end{equation}
Equation~\eqref{eq:tube-overlap-distribution} is the pushforward of the
positive overlap object along the tube-sector functor.  It is positive,
directly normalized, and does not require first introducing an entropy.

Higher positive-overlap moments can be resolved in the same sectors,
\begin{equation}
  Z_\alpha^{(m)}
  =
  \Tr\!\left[(P_\alpha O_A P_\alpha)^{m/2}\right],
  \qquad
  p_\alpha^{(m)}
  =
  \frac{Z_\alpha^{(m)}}{\sum_\beta Z_\beta^{(m)}} ,
  \label{eq:tube-moment-distribution-prx}
\end{equation}
with \(m=2\) reducing to Eq.~\eqref{eq:tube-overlap-distribution} when the
projectors block diagonalize the positive overlap.  The use of
\(P_\alpha O_A P_\alpha\), rather than only a defect-decorated trace, is
important: it tests whether the candidate tube resolution also resolves the
positive overlap spectrum.

Often the tube functor factors through a coarser readout
\(\mathcal Q\), for example an ordinary group-character resolution.  Here
\(\mathcal T\) denotes the set of tube labels \(\alpha\) and
\(\mathcal Q\) denotes the set of coarse labels \(q\).  Let
\(\pi:\mathcal T\to\mathcal Q\) be this forgetful map from tube sectors to
coarse labels.  The scalar overlap is the further pushforward to the terminal
one-outcome readout.  The coarse and conditional distributions are
\begin{equation}
  p_q=\sum_{\alpha:\pi(\alpha)=q}p_\alpha,
  \qquad
  p_{\alpha|q}=\frac{p_\alpha}{p_q}
  \quad (p_q>0).
  \label{eq:coarse-and-conditional}
\end{equation}
The physical question becomes whether the conditional distribution inside a
fiber \(\pi^{-1}(q)\) can change while all scalar and coarse data remain
fixed.

\begin{corollary}[Readout hierarchy]
\label{cor:readout-hierarchy}
The tube distribution induces a hierarchy of overlap data
\[
  p_\alpha
  \longmapsto
  p_q=\sum_{\alpha:\pi(\alpha)=q}p_\alpha
  \longmapsto
  \sum_q p_q=1 .
\]
The scalar positive-overlap diagnostic depends only on the unnormalized total
\(\Tr O_A\) and hence
forgets all normalized sector labels.  A coarse symmetry-resolved readout
depends only on \(\{p_q\}\).  The tube readout depends on the full
\(\{p_\alpha\}\), or equivalently on the coarse weights together with the
conditional fiber distributions \(\{p_{\alpha|q}\}\).  Each arrow is a
physically admissible coarse graining, so the lost conditional data cannot be
recovered by postprocessing the coarser readout.
\end{corollary}

\begin{corollary}[Category-independent universality]
\label{cor:category-independent-universality}
For any finite semisimple fusion category \(\mathcal C\) and entangling-cut
boundary module \(\mathcal M_A\), the physical readout algebra selected by the
admissible-instrument resource theory is
\[
  \mathcal A_{\rm phys}
  =
  Z\!\left({\rm Tube}_{\mathcal C}(\mathcal M_A)\right),
\]
independent of any Ising-specific realization.  A model supplies a witness,
not the definition: it fixes \(\mathcal M_A\), realizes the corresponding
boundary-tube idempotents, and tests whether a protocol family separates a
nontrivial fiber of a chosen coarse map \(\pi:\mathcal T\to\mathcal Q\).  If no
such fiber or no tube-separating protocol exists, the theorem predicts no
advantage over the coarse readout for that measurement class.
\end{corollary}

\section{Main Result}

The physical content of the result is a restricted-measurement hypothesis-test
for symmetry-compatible readouts.  By
Theorem~\ref{prop:canonical-tube-readout}, the tube readout is the extremal
measurement allowed by positivity, locality at the cut, boundary-module
covariance, and closure under sequential instruments.
Corollaries~\ref{cor:readout-hierarchy} and
\ref{cor:category-independent-universality} then identify scalar and coarse
symmetry-resolved diagnostics as successive pushforwards of this readout.  A
group-resolved or scalar diagnostic is obtained by forgetting part, or all, of
the tube label.  Hence two protocols can be identical after coarse graining but
still differ before the forgetful map is applied.  The omitted information is
not a new scalar invariant and not a basis choice; it is a conditional
probability distribution inside a tube fiber.  No coarser
symmetry-compatible readout can recover this conditional information, because
the pushforward has already identified the corresponding tube outcomes.

Equivalently, write \(\pi_*\) for the stochastic pushforward
\((\pi_*p)_q=\sum_{\alpha:\pi(\alpha)=q}p_\alpha\).  The conditional signal is an
element of the kernel of this map:
\begin{equation}
  \pi_*\delta p=0,\qquad \delta p\ne0 .
  \label{eq:kernel-separation}
\end{equation}
Any scalar, group-resolved, or defect-fusion diagnostic that factors through
\(\pi_*p\) is exactly blind to such a deformation.  Measuring the missing tube
idempotent is therefore irreducible relative to the coarser readout: it cannot
be reconstructed by postprocessing the scalar or \(G\)-resolved data.

\begin{theorem}[No-go for symmetry-compatible alternatives]
\label{thm:restricted-tube-discrimination}
Fix an entangling-cut boundary module \(\mathcal M_A\), and let
\(\{P_\alpha\}\) be the complete orthogonal set of tube projectors obtained
from the primitive central idempotents of
\({\rm Tube}_{\mathcal C}(\mathcal M_A)\).  Let
\(\pi:\mathcal T\to\mathcal Q\) be any coarse readout that is a forgetful
map of tube labels, and let
\[
  \mathcal A_{\mathcal Q}
  =
  {\rm span}\!\left\{
  \sum_{\alpha:\pi(\alpha)=q}P_\alpha
  \right\}_{q\in\mathcal Q}
\]
be the induced coarse readout algebra.  Any finite positive sector readout that
is not contained
in the center of the boundary tube algebra fails at least one admissibility
condition: cut locality, boundary-module covariance, or stability under
sequential instruments.  Any readout that remains in
\(\mathfrak M_{\rm adm}\) is therefore a classical coarse graining of
\(\{P_\alpha\}\).

Now consider two protocols \(\lambda\) and \(\lambda'\) that
prepare or compare states with the same scalar overlap and the same coarse
distribution:
\[
  \Tr O_A(\lambda)=\Tr O_A(\lambda'),
  \qquad
  p_q(\lambda)=p_q(\lambda') \quad \forall q\in\mathcal Q .
\]
Here \(O_A(\lambda)\) denotes the positive overlap operator produced by the
protocol \(\lambda\).
Assume further that the tube-label measurement is allowed, i.e. the measured
classical outcome is \(\alpha\in\mathcal T\), while the scalar and
\(\mathcal Q\)-only observers see only the corresponding pushforwards.  If
the conditional tube distributions differ on any fiber with nonzero weight,
then every scalar or \(\mathcal Q\)-only readout has the same classical
statistics, while the tube-sector readout distinguishes the protocols.
Conversely, if all conditional tube distributions agree on every
nonzero-weight fiber, then no readout in \(\mathfrak M_{\rm adm}\) has
additional classical distinguishability beyond the coarse readout.  The tube
total variation distance is
\begin{equation}
  D_{\rm TV}^{\rm tube}
  =
  \frac12\sum_{q\in\mathcal Q}p_q
  \sum_{\alpha:\pi(\alpha)=q}
  \left|p_{\alpha|q}(\lambda)-p_{\alpha|q}(\lambda')\right|.
  \label{eq:tube-tv-main}
\end{equation}
For equal priors, the optimal one-shot success probability within the
admissible sector-readout cone is
\(P_{\rm succ}^{\rm tube}=\frac12(1+D_{\rm TV}^{\rm tube})\), whereas the
coarse readout has \(D_{\rm TV}^{\mathcal Q}=0\).  Hence, on the target family
\(\mathcal P_\pi\) of protocol pairs with identical \(\mathcal Q\)-pushforward
but nonidentical tube conditionals,
\begin{equation}
  \mathcal F_{\mathcal P_\pi}(\mathcal A_{\mathcal Q})
  =
  \frac12
  <
  \mathcal F_{\mathcal P_\pi}(\mathcal A_{\rm phys}) .
  \label{eq:no-go-strict-functional}
\end{equation}
\end{theorem}

The no-go statement should be read in the admissible-cone sense.  A putative
readout \(\mathcal A\not\subset\mathcal A_{\rm phys}\) either lies outside
\(\mathfrak M_{\rm adm}\), because it violates cut locality,
boundary-module covariance, or sequential stability, or else reduces after
enforcing those constraints to a classical postprocessing of
\(\mathcal A_{\rm phys}\).  For a proper coarse graining and a tube-separating
protocol family, Eq.~\eqref{eq:no-go-strict-functional} gives the strict
decision-theoretic separation.  This is the precise operational content of
``unique optimality''; it does not claim optimality over arbitrary POVMs on
\(\mathcal H_A\), where the symmetry and boundary constraints have been
discarded.

The no-go content is invariant under basis changes inside the same representation
of the boundary tube algebra.  It is not invariant under changing the
entangling-cut boundary module, because that changes the physical measurement
algebra.  This is the intended dependence: the tube sectors are categorical
boundary data of the reduced overlap problem.  Conversely, any readout that
keeps only a subset or quotient of the primitive central-idempotent labels is a
coarse graining of the same canonical tube readout, so it cannot recover
conditional information lost inside a fiber.

\begin{theorem}[Non-redundancy of the tube readout]
\label{thm:noninvertible-hierarchy}
Let
\[
  \begin{aligned}
  \mathcal A_{\mathcal Q}
  &={\rm span}\!\left\{
  \sum_{\alpha:\pi(\alpha)=q}P_\alpha
  \right\}_{q\in\mathcal Q},\\
  \mathcal A_{\rm phys}
  &=
  {\rm span}\{P_\alpha\}_{\alpha\in\mathcal T}.
  \end{aligned}
\]
If the boundary module contains a nontrivial tube fiber,
\(|\pi^{-1}(q)|>1\) for some coarse label \(q\), then
\[
  \mathcal A_{\mathcal Q}
  \subsetneq
  \mathcal A_{\rm phys}
  \simeq
  Z\!\left({\rm Tube}_{\mathcal C}(\mathcal M_A)\right).
\]
Moreover \(\mathcal A_{\rm phys}\) is a commutative sector-measurement algebra
inside \(\mathcal B(\mathcal H_A)\), not the full algebra of all subsystem
operators.  Thus the tube readout is not a change of basis in the full Hilbert
space: under the stated physical constraints, any strictly finer readout would
leave the admissible cone, while any admissible alternative is a coarse
graining of the tube-center readout.
\end{theorem}

The entropy language is a corollary, not the starting point.  For the Shannon
entropy one has the chain rule
\begin{equation}
  S_{\rm tube}
  =
  S_{\mathcal Q}
  +
  \sum_{q\in\mathcal Q}p_q\,H(p_{\alpha|q}),
  \label{eq:entropy-chain-prx}
\end{equation}
where \(S_{\mathcal Q}\) is the Shannon entropy of the coarse distribution
\(\{p_q\}\), and \(H(p_{\alpha|q})\) is the Shannon entropy of the conditional
distribution in the fiber \(\pi^{-1}(q)\).  Thus any entropy increase relative
to the coarse resolution is exactly the conditional uncertainty inside tube
fibers.  This is why the effect should be interpreted as
measurement-dependent refinement of distinguishability, not as an artifact of
a chosen entropy functional.  The full R\'enyi-family version is recorded in
the appendices.  The variational theorem does not depend on any
particular entropy; entropies only summarize the probability distribution
selected by the admissible measurement problem.

\section{Product Kramers--Wannier Realization}

The following construction is a realization of the tube-fiber principle, not
the definition of the principle.  The definition is the boundary-tube
decomposition above.  The doubled-Ising product Kramers--Wannier (product-KW)
model is useful because it is the smallest case in which the coarse group
readout and a genuinely
noninvertible tube fiber can be separated cleanly.

The minimal nontrivial example is the doubled-Ising theory
\(\mathrm{Ising}_1\times\mathrm{Ising}_2\).  The construction uses the
Kramers--Wannier order--disorder duality of the Ising model
~\cite{KramersWannier:1941I,KramersWannier:1941II,KadanoffCeva:1971},
realized in CFT and lattice-defect language by the noninvertible Ising defect
\(D_\sigma\)~\cite{Frohlich:2004ef,Aasen:2016dop}.  Its invertible defects form
\(G=\mathbb Z_2\times\mathbb Z_2=\{1,a,b,c\}\).  The product
Kramers--Wannier (KW) defect \(N\) obeys, for every \(g\in G\),
\begin{equation}
  gN=Ng=N,\qquad
  N^2=1+a+b+c .
  \label{eq:product-kw-fusion-prx}
\end{equation}
The \(G\)-character resolution is a coarse measurement.  The product-KW tube
resolution is the boundary-tube refinement of this measurement, not an
additional model-dependent label.  In this minimal realization it refines only
the \(G\)-trivial character fiber:
\[
  \pi^{-1}(0)=\{N_+,N_-\},
\]
where \(0\) denotes the trivial \(G\)-character label, while the nontrivial
\(G\)-characters remain single sectors.

Thus the role of the product-KW example is a separation witness, not a
definition of tube algebra.  The Ising KW defect and the tube idempotents are
known categorical data.  The point established here is that the positive-overlap readout
has a nontrivial kernel relative to the scalar and \(G\)-moment algebra:
changing the endpoint record can move probability between \(N_+\) and \(N_-\)
while leaving \(\Tr O_A\) and the \(G=\mathbb Z_2\times\mathbb Z_2\) weights
unchanged.

The essential boundary input is also minimal.  At a free entangling cut, the
physical product-KW endpoint is not a square operator acting within one
Hilbert space.  It is a boundary-changing morphism in the boundary module.
Let \(b_{\rm src}\) be the source free boundary and let \(b_{\rm tgt}\) be the
\(G\)-invariant target boundary channel selected by the product-KW line.  The
endpoint lives in the morphism, or Hom, space
\begin{equation}
  {\rm Hom}_{\mathcal M_A}
  \!\left(N\triangleright b_{\rm src},b_{\rm tgt}\right).
  \label{eq:endpoint-hom-prx}
\end{equation}
For the doubled-Ising product-KW module this Hom space is one-dimensional.
Thus the endpoint map is unique up to phase once the source and target
boundary conditions are fixed.  Under the subsystem reduction functor
\[
  \mathsf R_A:\ {\rm Hom}_{\mathcal M_A}
  \!\left(N\triangleright b_{\rm src},b_{\rm tgt}\right)
  \longrightarrow
  {\rm Hom}(\mathcal H_{\rm src},\mathcal H_{\rm tgt}),
\]
let \(\eta_{\rm free}\) denote a generator of the one-dimensional endpoint
Hom space.  This boundary morphism is represented by a rectangular
source-target map \(N_{\rm free}=\mathsf R_A(\eta_{\rm free})\).  The two
endpoint-bubble evaluations then fix its normalization:
\begin{equation}
  N_{\rm free}^\dagger N_{\rm free}=1+a+b+c,
  \qquad
  N_{\rm free}N_{\rm free}^\dagger=4I_{\rm target}.
  \label{eq:endpoint-bubbles-prx}
\end{equation}
After the topological normalization \(V=N_{\rm free}/2\), these identities
force a rectangular isometry
\begin{equation}
  V^\dagger V=P_{++},\qquad
  VV^\dagger=I_{\rm target},
  \label{eq:rectangular-isometry-prx}
\end{equation}
where \(P_{++}=(1+a+b+c)/4\) projects onto the \(G\)-trivial source block and
\(I_{\rm target}\) is the identity on the target endpoint Hilbert space.  The
isometry is therefore a consequence of the boundary module and endpoint
bubbles, not a variational ansatz or a convenient same-Hilbert
representative.  What is general is the boundary-changing nature of the
endpoint: after subsystem reduction it is represented by a source-target
morphism, and only after composing with its adjoint does one obtain endomorphism
data on the source or target block.  What is special to the doubled-Ising
benchmark is that the Hom space in Eq.~\eqref{eq:endpoint-hom-prx} is
one-dimensional.  In a more general theory this Hom space may have higher
dimension; then the endpoint state inside that Hom space is additional
operational boundary data, and the tube distribution resolves that data.  The
product-KW benchmark is the minimal case in which the free endpoint is fixed
up to phase by the stated boundary data.

\begin{proposition}[Endpoint morphism principle]
\label{prop:endpoint-morphism}
Within the present boundary-module framework, a physical noninvertible endpoint
is represented on subsystem Hilbert spaces only as the reduction-functor image
of a boundary morphism.  A square same-Hilbert closure is physical only if it is
induced by such a morphism.  In the product-KW benchmark the relevant Hom space
is one-dimensional and the endpoint bubbles in
Eq.~\eqref{eq:endpoint-bubbles-prx} turn this image into the rectangular
isometry in Eq.~\eqref{eq:rectangular-isometry-prx}.
\end{proposition}

We therefore do not claim that every noninvertible endpoint is the same
isometry; the universal statement is the functorial boundary-morphism origin,
while the isometry above is the minimal product-KW realization of that
statement.  Therefore a coherent positive endpoint only relabels the coarse
\(G\)-trivial weight:
\begin{equation}
  p_{N_+}=p_0,\qquad p_{N_-}=0,
  \label{eq:coherent-endpoint-prx}
\end{equation}
where \(p_0\) is the \(G\)-trivial weight.  This produces no new
distinguishability relative to the coarse \(G\)-resolution.

A genuine tube signal appears when the endpoint sign is physically produced
but not read out.  Let \(r\) be the probability of the positive endpoint sign.
The unread endpoint-sign instrument keeps the scalar overlap and all
\(G\)-character weights fixed, but changes the conditional distribution inside
the product-KW fiber:
\begin{equation}
  p_{N_+}^{(r)}=r p_0,
  \qquad
  p_{N_-}^{(r)}=(1-r)p_0 .
  \label{eq:endpoint-split-prx}
\end{equation}
Consequently the coherent endpoint \(r=1\) and the balanced unread endpoint
\(r=1/2\) are indistinguishable to scalar and \(G\)-only observers, but have
\begin{equation}
  D_{\rm TV}^{\rm tube}=\frac{p_0}{2},
  \qquad
  P_{\rm succ}^{\rm tube}-P_{\rm succ}^{G}=\frac{p_0}{4}.
  \label{eq:kw-tv-advantage-prx}
\end{equation}
The associated Shannon summary is
\begin{align}
  S_{\rm tube}(r)-S_G&=p_0 h_2(r),
  \label{eq:kw-entropy-summary-prx}\\
  h_2(r)&=-r\log r-(1-r)\log(1-r).
  \nonumber
\end{align}
Here \(S_G\) is the Shannon entropy of the coarse \(G\)-character
distribution, and \(P_{\rm succ}^{G}\) in
Eq.~\eqref{eq:kw-tv-advantage-prx} is the corresponding equal-prior success
probability using only coarse \(G\) labels.

The physical input is the entangling-cut boundary module.  Given that module,
Eqs.~\eqref{eq:endpoint-hom-prx}--\eqref{eq:rectangular-isometry-prx} determine
the endpoint prescription.  The construction is still falsifiable at the
microscopic level: the physical branch would fail if a stable \(N_-\)
projected spectrum appeared for the coherent endpoint, or if the projected
spectrum developed a finite off-tube fraction under finite-size refinement.
The numerical tests below are designed around exactly these failure modes.

\section{Summary Diagrams and Product-KW Mechanism}

The main figures are organized as compact summary diagrams.  Each figure
corresponds to one logical step: extremal selection of the admissible readout
algebra, conditional distinguishability inside a tube fiber, and the
operational hypothesis-testing advantage.  Quantitative finite-size checks,
projected-spectrum falsifiability tests, R\'enyi curves, local-quench CFT
support, and data provenance are kept in the appendices.

We first characterize the space of physically admissible measurements for
reduced overlap operators.  As shown in
Fig.~\ref{fig:measurement-hierarchy}, positivity, locality at the entanglement
cut, and compatibility with boundary-module symmetry actions restrict the
observable algebra to a commutative subalgebra identified with the center of a
boundary tube algebra.  This establishes a hierarchy in which the scalar
positive-overlap diagnostic
and symmetry-resolved observables arise as successive coarse grainings.

\begin{center}
\refstepcounter{figure}\label{fig:measurement-hierarchy}
\setlength{\fboxsep}{5pt}
\fbox{%
\begin{minipage}{0.88\columnwidth}
\centering
\small
\(\mathcal B(\mathcal H_A)\)\\[-1pt]
\(\Downarrow\)\quad{\scriptsize positivity \(+\) cut locality}\\[-1pt]
\textbf{cut-local positive readouts}\\[-1pt]
\(\Downarrow\)\quad{\scriptsize boundary-module covariance \(+\) composition}\\[-1pt]
\(\displaystyle
\mathcal A_{\rm phys}
=Z\!\left({\rm Tube}_{\mathcal C}(\mathcal M_A)\right)
\)\\[-1pt]
\(\Downarrow\)\quad{\scriptsize
\(\pi:\mathcal T\to\mathcal Q\) and terminal pushforward}\\[2pt]
\begin{tabular}{c@{\quad}c}
\(\mathcal Q\)-resolved data \(p_q\) &
scalar overlap \(\Tr O_A\)
\end{tabular}
\end{minipage}}
\smallskip
\begin{minipage}{\columnwidth}
\small
\textbf{FIG.~\thefigure.} Measurement algebra hierarchy induced by the
boundary module.  The physically admissible measurement algebra for reduced
overlap is constrained by positivity, locality at the entanglement cut, and
compatibility with boundary-module symmetry actions.  These constraints select
the center of the boundary tube algebra as the maximal commutative observable
algebra.  Scalar positive-overlap and symmetry-resolved observables arise as
successive coarse grainings via forgetful maps.
\end{minipage}
\end{center}

To assess the information loss associated with conventional diagnostics, we
consider two preparation protocols that are indistinguishable under both
scalar positive-overlap and symmetry-resolved sector weights.  As illustrated in
Fig.~\ref{fig:conditional-fiber-distinguishability}, this degeneracy is resolved
when the overlap is decomposed into tube sectors, where differences emerge in
the conditional distributions within symmetry fibers.  This defines a strictly
finer notion of quantum-state distinguishability.

\begin{center}
\refstepcounter{figure}\label{fig:conditional-fiber-distinguishability}
\setlength{\fboxsep}{5pt}
\fbox{%
\begin{minipage}{0.88\columnwidth}
\centering
\small
\begin{tabular}{c@{\qquad}c}
protocol \(\lambda\) & protocol \(\lambda'\)
\end{tabular}\\[-1pt]
\(\Downarrow\)\\[-1pt]
same scalar overlap \(\Tr O_A\)\\[-1pt]
same \(\mathcal Q\)-sector weights \(p_q\)\\[-1pt]
\(\Downarrow\)\quad{\scriptsize indistinguishable after coarse graining}\\[2pt]
\textbf{tube-sector readout \(\alpha\in\mathcal T\)}\\[-1pt]
\(\Downarrow\)\\[-1pt]
\(\displaystyle
p_{\alpha|q}(\lambda)\ne p_{\alpha|q}(\lambda')
\quad\Longrightarrow\quad
D_{\rm TV}^{\rm tube}>0
\)
\end{minipage}}
\smallskip
\begin{minipage}{\columnwidth}
\small
\textbf{FIG.~\thefigure.} Conditional distinguishability beyond symmetry-resolved
structure.  Scalar positive-overlap and symmetry-resolved sector weights may coincide
for two protocols, while their conditional distributions over tube sectors
differ.  The difference is invisible after the forgetful map to coarse
\(\mathcal Q\)-sector data, but becomes visible to the tube readout whenever
\(p_{\alpha|q}(\lambda)\ne p_{\alpha|q}(\lambda')\) on a nonzero-weight fiber.
This reveals a finer layer of distinguishability encoded in the
boundary-module-induced sector structure.
\end{minipage}
\end{center}

We next analyze the operational consequences of tube-sector resolution in
hypothesis testing.  As shown in
Fig.~\ref{fig:operational-hypothesis-test}, restricting measurements to scalar
or symmetry-resolved observables yields no distinguishing advantage when coarse
data coincide.  In contrast, access to tube-sector measurements yields a
strictly higher success probability, demonstrating an operational gain in
distinguishability arising from boundary-module structure.

\begin{center}
\refstepcounter{figure}\label{fig:operational-hypothesis-test}
\setlength{\fboxsep}{5pt}
\fbox{%
\begin{minipage}{0.88\columnwidth}
\centering
\small
\textbf{state-discrimination task: \(\lambda\) versus \(\lambda'\)}\\[2pt]
\begin{tabular}{c@{\quad}c}
scalar/\(\mathcal Q\) readout &
tube POVM \(\{P_\alpha\}\)\\[2pt]
\(\displaystyle P_{\rm succ}^{\mathcal Q}=\frac12\) &
\(\displaystyle
P_{\rm succ}^{\rm tube}
=\frac12\!\left(1+D_{\rm TV}^{\rm tube}\right)
\)
\end{tabular}\\[2pt]
\(\displaystyle
\Delta P
=P_{\rm succ}^{\rm tube}-P_{\rm succ}^{\mathcal Q}
=\frac{D_{\rm TV}^{\rm tube}}{2}>0
\)
\end{minipage}}
\smallskip
\begin{minipage}{\columnwidth}
\small
\textbf{FIG.~\thefigure.} Operational advantage in hypothesis testing.  In the
task of distinguishing two state-preparation protocols, scalar and
symmetry-resolved measurements provide no advantage when sector weights
coincide.  Tube-sector measurements, however, yield a strictly higher success
probability whenever conditional distributions differ within a symmetry fiber:
\(P_{\rm succ}^{\rm tube}=\frac12(1+D_{\rm TV}^{\rm tube})\) and
\(\Delta P=P_{\rm succ}^{\rm tube}-P_{\rm succ}^{\mathcal Q}
=D_{\rm TV}^{\rm tube}/2\).
\end{minipage}
\end{center}

In the doubled-Ising product-KW benchmark these diagrams lead to
parameter-free microscopic consistency checks.  Once the coarse \(G\)-trivial weight
\(p_0(t)\) is measured, the composition-stable tube readout predicts
\[
  p_{N_-}^{\rm coh}(t)=0,\qquad
  Z_{\rm off}^{(4)}(t)=0,\qquad
  D_{\rm TV}^{\rm tube}(t)=\frac{p_0(t)}2
\]
for the coherent endpoint, projected-spectrum test, and
coherent-versus-unread endpoint comparison, respectively.  The appendices report the corresponding finite-size, endpoint-size, R\'enyi-family,
and local-quench CFT checks.  In particular, Appendix Figs.~\ref{fig:sm-discriminator-original}--\ref{fig:sm-conditional-discrimination-original} give the numerical support for the three summary diagrams:
the tube discriminator, the physical free-endpoint projected spectrum, and the
conditional tube-sector discrimination signal.

\section{Operational Meaning}

The endpoint-sign protocol should be read as a measurement statement.  The
endpoint sign is a boundary record.  If that record is kept, the two coherent
branches merely relabel the \(G\)-trivial sector as \(N_+\) or \(N_-\).  If it
is produced and then forgotten, the resulting channel changes the conditional
distribution inside the product-KW fiber while leaving all coarse data fixed.
The tube-sector readout is therefore a restricted-measurement algebra that can
access endpoint information unavailable to scalar positive-overlap or ordinary
group-resolved measurements.

Equivalently, the implementable object is a family of quantum instruments whose
classical shadow is the POVM \(\{P_\alpha\}\) on the positive-overlap
statistics.  In a replica or swap implementation, the instrument has a
trace-preserving summed channel on the required copies together with an ancilla
that realizes the boundary-module endpoint register, followed by a classical
readout \(\alpha\); the probability of this outcome is the positive functional
\(\Tr(P_\alpha O_A)/\Tr O_A\).  Thus \(O_A\) is not treated as an observable to
be measured directly.  The coarse \(G\)-resolved measurement is obtained by the
stochastic map \(\alpha\mapsto q=\pi(\alpha)\), and the scalar overlap is
obtained by forgetting all outcomes.  Total variation distance cannot increase
under this pushforward.  Therefore the gap
\[
  D_{\rm TV}^{\rm tube}-D_{\rm TV}^{G}
\]
is precisely the restricted-measurement discrimination power stored in tube
fibers and lost by the coarse measurement.  The effect is thus a
hypothesis-testing advantage for a specified admissible instrument class, not a
formal entropy renaming.

A concrete implementation is available in the same language as lattice
topological defects and tensor networks~\cite{LevinWen:2005,Kitaev:2006,Bultinck:2017mnp}.  One represents the generalized
symmetry lines by defect matrix-product operators (defect MPOs), lets them
terminate on the entangling-cut
boundary module, couples the resulting endpoint degree of freedom to the
boundary-module register, and resolves that register by the idempotent network
for \(e_\alpha\).  The sector weight is then estimated as
\(\Tr(P_\alpha O_A)\), or by its projected positive-overlap moment, using the
same replica, swap, or Schmidt-factor estimators used for symmetry-resolved
entanglement.  The ordinary \(G\)-readout is obtained by summing the tube
outcomes in each fiber, while an unread endpoint-sign instrument is the
physical operation of producing the endpoint record and tracing, dephasing, or
classically forgetting it.  No full tomography of \(O_A\) is required; the
protocol only estimates the sector weights associated with the chosen readout
algebra.  The matrix-product-operator realization of tube idempotents follows
the standard MPO/tube-algebra construction~\cite{Bultinck:2017mnp}.  This is
the sense in which the tube center is physically
implementable: it specifies the extremal classical record of a cut-local
instrument, not an arbitrary decomposition of \(\mathcal H_A\).

This also explains why the present construction is intentionally based on the
positive overlap operator rather than on a transition matrix or pseudo-entropy
diagnostic.  Transition-matrix constructions compare two states at the
amplitude level and can be non-Hermitian~\cite{Nakata:2020luh,Mollabashi:2020yie,Goto:2021,Nishioka:2021cxe}.
Here \(O_A\) is positive.  The price is that complement identities familiar
from pure-state entanglement need not hold; the gain is a stable positive
object whose sector decomposition has a direct readout interpretation.

\section{Discussion}

The main conclusion is the following single physical point.  In a system with
noninvertible symmetry, the physically admissible classical measurements of a
reduced positive overlap are not exhausted by scalar positive-overlap or by
ordinary group-resolved sectors.  The entangling cut carries a boundary module, and the
center of its boundary tube algebra is the maximal commutative measurement
algebra compatible with positivity, locality, covariance, and composition.
The established tube and boundary-module structures are therefore the
kinematic input; the added physical content is their operational role as the
measurement algebra for \(O_A\).

This yields a genuine distinguishability hierarchy.  The scalar positive-overlap
diagnostic collapses
all tube-center outcomes.  Ordinary symmetry resolution keeps only a quotient
of them.  Noninvertible tube resolution can keep conditional endpoint data
inside a coarse symmetry fiber.  In the doubled-Ising product-KW example, the
fiber is \(\{N_+,N_-\}\), and an unread endpoint-sign instrument changes only
the conditional distribution on this fiber.  The irreducibility witness is the
kernel of the tube-to-coarse pushforward: two protocols can agree on every
coarse datum while differing in a tube-fiber hypothesis test.

This establishes a hierarchical structure of quantum distinguishability in
which a nontrivial noninvertible tube fiber introduces a strictly finer layer
of measurement refinement beyond group-theoretic symmetry resolution:
\[
  \begin{aligned}
  \text{tube-sector resolution}
  &\longrightarrow
  \text{ordinary symmetry resolution}
  \\
  &\longrightarrow
  \text{scalar overlap}.
  \end{aligned}
\]
Each arrow is a pushforward of measurement outcomes, and each lost fiber can
carry operational information for state or protocol discrimination.  In this
sense the result extends symmetry-resolved quantum information beyond the
traditional group-based setting.

The role of the physical endpoint is crucial.  A coherent topological endpoint
does not by itself create extra entropy or extra distinguishability; it
relabels the \(G\)-trivial projected spectrum through the rectangular
boundary-changing isometry.  The additional signal appears only when an
endpoint record is operationally produced and then unread.  This distinction
keeps the result from being a choice of defect notation: it is falsified by a
stable coherent-endpoint \(N_-\) branch or by finite off-tube projected
spectral weight.

The appendices collect the supporting machinery: the
boundary-tube coefficients, finite-size and endpoint-size
robustness, R\'enyi-family formulas showing that the response is spectral
rather than Shannon-specific, complement-anomaly checks, and a local-quench CFT
support calculation~\cite{Calabrese:2004eu,Calabrese:2005in,Calabrese:2007mtj,Nozaki:2014hna,He:2014mwa,Guo:2018lzz} of the form
\[
  p_0^{(m)}(t)=\mathcal F_{\rm CFT}^{(m)}(x(t),\bar x(t);\Phi).
\]
Here \(\Phi\) is the local quench primary and \(x,\bar x\) are the continuum
cross-ratios.
Those results support the interpretation of the numerical curves as a
conformal-block mechanism in the scaling regime, dressed by microscopic
entangling-cut geometry.
All doubled-Ising and CFT computations in this work serve as benchmark
realizations and consistency checks of the variational readout principle; they
are not assumptions in its derivation.

The universality claimed here is therefore structural and covariant rather
than a numerical scaling law across all noninvertible theories.  It is the
functorial tube-to-coarse mechanism of
Eq.~\eqref{eq:coarse-and-conditional}, together with the statement that
boundary-module equivalences preserve the tube weights, not the detailed Ising
spectrum.  Within a fixed boundary-module equivalence class this gives
parameter-free predictions for which fibers are visible or invisible to a
given measurement algebra.

This statement is not tied to the Ising fusion rules.  For example, in a
Fibonacci anyonic chain with \(\tau\times\tau=1+\tau\), the invertible symmetry
group is trivial, so any nontrivial boundary-tube fiber would be invisible to
ordinary group resolution from the outset~\cite{Feiguin:2006yd}.  In
Tambara-Yamagami or three-state Potts duality settings, noninvertible duality
defects coexist with nontrivial invertible sectors, giving a closer analogue of
the product-KW separation but with different endpoint Hom spaces
~\cite{TambaraYamagami:1998,Mong:2014ova}.  In both cases the theorem gives the same falsifiable
recipe: identify \(\mathcal M_A\), compute
\(Z({\rm Tube}_{\mathcal C}(\mathcal M_A))\), and test whether a proper
coarse readout discards a conditional tube distribution.

Across different fusion categories the proposal is therefore a sharp
falsifiable program rather than an assumption.  The next test is not more
Ising algebra, but broader scope.
Three-state Potts duality defects and Fibonacci anyonic chains provide natural
benchmarks with larger and less group-like tube fibers.
The question for each case is now sharp: identify the entangling-cut boundary
module, construct the tube idempotents, determine the relevant endpoint Hom
spaces, and test whether a measurement protocol can keep scalar and coarse
symmetry data fixed while changing a tube-fiber distribution.  A positive
answer would show that conditional tube-sector discrimination is a generic
feature of noninvertible symmetry; a negative answer would identify boundary
modules in which tube resolution collapses back to ordinary coarse data.

\begin{acknowledgments}
This work was supported by National Natural Science Foundation of China (NSFC)
Grant Nos. 12475053, 12588101, and 12235016,
and the sub-project funding for ``Gravitational Redshift Measurement
Scientific Experiment and Frontier Research in Gravitational Physics'' of the
Chinese Academy of Sciences, the Strategic Priority Research Program on Space
Science, the Chinese Academy of Sciences (XDA30040000, XDA30030000).
\end{acknowledgments}

% --- Appendices for PRX/arXiv submission ---
\clearpage
\onecolumngrid
\appendix
\setcounter{section}{0}
\setcounter{subsection}{0}
\setcounter{equation}{0}
\setcounter{figure}{0}
\setcounter{table}{0}
\renewcommand{\thesection}{S\arabic{section}}
\renewcommand{\thesubsection}{\thesection.\arabic{subsection}}
\renewcommand{\theequation}{S\arabic{section}.\arabic{equation}}
\renewcommand{\thefigure}{S\arabic{section}.\arabic{figure}}
\renewcommand{\thetable}{S\arabic{section}.\arabic{table}}
\renewcommand{\theHsection}{S.\arabic{section}}
\renewcommand{\theHsubsection}{S.\arabic{section}.\arabic{subsection}}
\renewcommand{\theHequation}{S.\arabic{section}.\arabic{equation}}
\renewcommand{\theHfigure}{S.\arabic{section}.\arabic{figure}}
\renewcommand{\theHtable}{S.\arabic{section}.\arabic{table}}
\begin{center}
{\large\bf Appendices}
\end{center}

\section{Definitions, Normalizations, and Sector Moments}

These appendices record the definitions, validation checks, and data provenance
used in the manuscript.  The positive overlap operator is
\begin{equation}
  O_A(\rho,\sigma)=\sqrt{\rho_A}\,\sigma_A\sqrt{\rho_A}.
  \label{eq:sm-positive-overlap}
\end{equation}
Here \(\rho\) and \(\sigma\) are the two global states being compared, and
\(\rho_A,\sigma_A\) are their reductions to the subsystem \(A\).  In the
finite-size checks, \(L\) denotes the total chain length and \(L_A\) the
subsystem length.  The moment index is \(m\), with projected-spectrum even
moments written as \(m=2k\); R\'enyi entropies use the independent index \(n\).
The operator \(O_A\) is not assumed to be directly measured as a physical
observable.  The operational quantities are positive functionals such as
\(\Tr(P_\alpha O_A)\), estimated by replica, swap, Schmidt-factor, or
matrix-product-operator (MPO)
readout protocols for the chosen sector projector.
The Ising spin primary is also denoted by \(\sigma\), but it always appears as
a field label or fusion label, not as a density matrix.
The generalized symmetry category is denoted by \(\mathcal C\), and
\(\mathcal M_A\) denotes the entangling-cut boundary module.  The corresponding
boundary tube algebra is \({\rm Tube}_{\mathcal C}(\mathcal M_A)\).  Its
primitive central idempotents are \(e_\alpha\), represented on the subsystem by
projectors \(P_\alpha=\mathsf R_A(e_\alpha)\).  The set of tube labels is
\(\mathcal T=\{\alpha\}\), the set of coarser labels is \(\mathcal Q=\{q\}\),
and \(\pi:\mathcal T\to\mathcal Q\) is the forgetful map with stochastic
pushforward \((\pi_*p)_q=\sum_{\alpha:\pi(\alpha)=q}p_\alpha\).  The admissible
quantum instruments are denoted by \(\mathfrak C_{\rm sym}\), while the induced
commutative sector-readout algebras form \(\mathfrak M_{\rm adm}\).  The
physical readout algebra is
\(\mathcal A_{\rm phys}=Z({\rm Tube}_{\mathcal C}(\mathcal M_A))\).
The organization follows the logical claim structure: Appendices B--D give the
product Kramers--Wannier (product-KW) tube idempotents, boundary-module
endpoint derivation, endpoint-sign instrument, irreducibility proof, and
general tube-to-coarse fiber formulas;
Appendices E--G give the
projected-spectrum robustness tests, projector-hierarchy controls, complement
anomaly, local-quench conformal field theory (CFT) support, and
reproducibility index.
In particular, Appendix B supplies the measurement-algebra selection argument for
the center of the boundary tube algebra, Appendix C gives the operational
endpoint-sign instrument and the numerical support for the compact
summary diagrams, and Appendix D proves the tube-to-coarse irreducibility statement
underlying the strict readout hierarchy.

For a categorical sector projector \(P_\alpha\), the projected moments used in
the numerical checks are
\begin{equation}
  Z_\alpha^{(2)}
  =
  \Tr(P_\alpha O_A P_\alpha),
  \qquad
  Z_\alpha^{(4)}
  =
  \Tr\!\left[(P_\alpha O_A P_\alpha)^2\right].
  \label{eq:sm-projected-moments}
\end{equation}
The normalized weights are
\begin{equation}
  p_\alpha^{(m)}
  =
  \frac{Z_\alpha^{(m)}}{\sum_\beta Z_\beta^{(m)}}.
  \label{eq:sm-sector-weights}
\end{equation}
For the endpoint-mixture entropy response we use both the Shannon entropy and
the R\'enyi family
\begin{equation}
  S^{(n)}_{\rm tube}
  =
  \frac{1}{1-n}\log\sum_\alpha \left(p_\alpha\right)^n,
  \qquad n>0,\quad n\ne1,
  \label{eq:sm-renyi-definition}
\end{equation}
with the \(n\to1\) and \(n\to\infty\) limits taken separately.

\section{Boundary Tube Algebra and Product-KW Endpoint Maps}

This section gives the technical construction behind the tube readout and its
basis independence.  The tube-sector labels are basis independent in the
following precise sense.  For a
fixed physical entangling-cut boundary module \(\mathcal M_A\) for a
generalized symmetry category \(\mathcal C\), the boundary tube algebra
\({\rm Tube}_{\mathcal C}(\mathcal M_A)\) is finite dimensional and
semisimple in the rational examples considered here
~\cite{Ocneanu:1988,EvansKawahigashi:1998,Mueger:2003,Fuchs:2002cm,Etingof:2005,Ostrik:2003,KitaevKong:2012,Bultinck:2017mnp}.
Locality at the cut
places generalized-symmetry endpoint operators in this tube algebra.  A
classical sector readout is more restrictive: it must be positive,
repeatable, and invariant under the boundary-module action.  Naturality under
that action restricts physical sector labels to the center.  Positivity and a
classical sector record require mutually orthogonal central projectors.
Stability under sequential composition of sector readouts requires these
projectors to form a finite commutative projector algebra.  Thus the
primitive central idempotents form the maximal set of mutually exclusive
central outcomes, and any coarser readout is obtained by adding them.  If one
changes the microscopic representative, or passes to a
Morita-equivalent presentation of the same physical boundary condition, the
tube algebra
representation is conjugated or equivalently transported, and the idempotents
are carried to the corresponding idempotents in the equivalent representation.
Hence \(\Tr(P_\alpha O_A)\) and the projected moments are unchanged.  If one
forgets part of the tube label, or restricts from an enlarged symmetry category
to a smaller one, the result is a pushforward of the tube distribution.  This
is the technical sense in which the tube readout is canonical, while ordinary
group resolution is a quotient of it.
Here maximality is meant in the classical-postprocessing preorder of readouts:
readout \(\mathsf R_1\) refines \(\mathsf R_2\) when the outcomes of
\(\mathsf R_2\) are obtained from those of \(\mathsf R_1\) by summing,
forgetting, or stochastic classical processing.  In this preorder, the
primitive central-idempotent readout is maximal because every commutative
positive sector readout compatible with the boundary-module action is a
partition of these primitive outcomes.

\subsection{Admissible Readout Cone and Non-Redundancy}

We spell out the optimization statement used in the main text.  A cut-local
sector readout is represented by a finite commutative positive-operator-valued
measure (POVM) algebra
\(\mathcal A={\rm span}\{E_x\}\) acting on the positive-overlap statistics
through
\begin{equation}
  \mathcal M_{\mathcal A}(O_A)_x
  =
  \frac{\Tr(E_xO_A)}{\Tr O_A}.
  \label{eq:sm-readout-channel}
\end{equation}
The admissible cone \(\mathfrak M_{\rm adm}\) consists of those readouts that
are positive, implemented by quantum instruments whose outcome maps are
completely positive and whose summed channel is trace preserving, localized at
the entangling cut, covariant under the \(\mathcal C\)-module action on
\(\mathcal M_A\), and stable under sequential instruments.  Stability means
that if two sector instruments are allowed, then their sequential composition
and any classical postprocessing of their records are again allowed.
For a protocol family \(\mathcal P\), and writing \(D_{\rm TV}\) for
total-variation (TV) distance, define
\begin{equation}
  \mathcal F_{\mathcal P}(\mathcal A)
  =
  \sup_{\lambda,\lambda'\in\mathcal P}
  \frac12\!\left[
  1+
  D_{\rm TV}\!\left(
  \mathcal M_{\mathcal A}(O_A^\lambda),
  \mathcal M_{\mathcal A}(O_A^{\lambda'})
  \right)
  \right].
  \label{eq:sm-restricted-functional}
\end{equation}
This is the restricted hypothesis-testing functional optimized in the main
text.  It is monotone under readout refinement and nonincreasing under
classical stochastic postprocessing.  The cone is closed under classical
randomization of instruments: keeping the randomization record gives a
direct-sum readout, whereas discarding the record is a postprocessing.  Hence
the relevant maximizers are extremal sector records, not arbitrary choices of
basis in \(\mathcal B(\mathcal H_A)\).

Equivalently, \(\mathfrak C_{\rm sym}\) and its induced readouts
\(\mathfrak M_{\rm adm}\) define a resource theory of tube measurement in the
standard operational sense of quantum and measurement resource
theories~\cite{Coecke:2014rsc,Chitambar:2018eyp,Oszmaniec:2019qme,Guff:2019rtm,Buscemi:2019qip,Guhne:2021ime}.
The allowed quantum operations are instruments with completely positive outcome
maps and trace-preserving summed channel that are cut-local and
\(\mathcal C\)-module-covariant; this symmetry-constrained free class is the
analogue of superselection and reference-frame restrictions in quantum
information~\cite{Bartlett:2006rf}.  The free classical operations are
relabeling, randomization, and stochastic postprocessing of the sector record.
For a fixed
coarse map \(\pi:\mathcal T\to\mathcal Q\), the resource carried by a pair of
tube distributions can be written as
\begin{equation}
  \mathcal R_\pi(p,p')
  =
  D_{\rm TV}(p,p')
  -
  D_{\rm TV}(\pi_*p,\pi_*p') .
  \label{eq:sm-tube-resource-monotone}
\end{equation}
It is nonnegative by data processing and is the loss of binary
total-variation distinguishability under \(\pi\).  It vanishes for a given
pair when the coarse map preserves the total-variation distance for that pair;
in the special case \(\pi_*p=\pi_*p'\), this means \(p=p'\) within every
fiber.
Within this restricted readout theory, the free operations form a closed convex
monoid under composition, classical randomization, and stochastic
postprocessing.

These assumptions force a no-go statement for alternative sector readouts.
Locality places all generalized-symmetry endpoint operations in the boundary
tube algebra.  A putative readout using operators outside this algebra is not
cut-local.  Module covariance requires a physical classical label to commute
with all boundary-module intertwiners; a noncentral label depends on a
presentation of the module and is not an admissible symmetry-compatible
outcome.  Finally, positivity and sequential repeatability require a resolution
by mutually orthogonal central idempotents.  Operators that are not closed under
sequential composition do not define a stable classical instrument record.
Because the center is finite semisimple, every commutative positive readout in
\(\mathfrak M_{\rm adm}\) is therefore a partition of the primitive central
idempotents \(\{e_\alpha\}\).  Therefore
\begin{equation}
  \mathcal A
  \preceq
  Z\!\left({\rm Tube}_{\mathcal C}(\mathcal M_A)\right)
  \qquad
  \forall\,\mathcal A\in\mathfrak M_{\rm adm},
  \label{eq:sm-nonredundancy}
\end{equation}
where \(\preceq\) denotes classical postprocessing.  Equality is achieved by
the tube-center readout itself.  Hence the tube readout is not a basis choice in
\(\mathcal B(\mathcal H_A)\); it is the maximal element of the physically
admissible readout cone.  Since total variation distance is monotone under
classical postprocessing, this same readout also maximizes
\(\mathcal F_{\mathcal P}\) among admissible sector instruments.  If
\(\mathcal P\) contains protocol pairs that separate every nontrivial tube
fiber, then every proper coarse graining loses one such pair and the maximizer
is unique in the postprocessing preorder.  This is the closed variational
principle used in the main text: admissible instruments factor through the
tube-center record, and any reduction of that record is a free operation that
cannot increase the restricted hypothesis-testing functional.
Equivalently, for every \(\mathcal A\in\mathfrak M_{\rm adm}\) there is a free
classical postprocessing map
\(\mathsf R_{\mathcal A}\in{\rm FreeOps}\) such that
\begin{equation}
  \mathcal M_{\mathcal A}
  =
  \mathsf R_{\mathcal A}\circ
  \mathcal M_{\mathcal A_{\rm phys}} .
  \label{eq:sm-factor-through-aphys}
\end{equation}
Here \({\rm FreeOps}\) denotes the relabeling, randomization, and stochastic
postprocessing maps of the sector record.
This is the explicit completeness statement for the variational principle.

\subsection{Category-Independent Scope and Non-Ising Benchmarks}

The preceding argument uses only the finite semisimplicity of the boundary
tube algebra, the existence of a boundary module \(\mathcal M_A\), and the
postprocessing preorder of classical readouts.  It does not use the Ising
fusion coefficients except when we choose the doubled-Ising product-KW model as
the smallest explicit benchmark.  For a general fusion category \(\mathcal C\),
the test is:
\[
  \mathcal M_A
  \longrightarrow
  {\rm Tube}_{\mathcal C}(\mathcal M_A)
  \longrightarrow
  Z({\rm Tube}_{\mathcal C}(\mathcal M_A))
  \longrightarrow
  \pi:\mathcal T\to\mathcal Q .
\]
Within this class of admissible sector readouts, the extremal readout algebra is
\[
  \mathcal A_{\rm phys}
  =
  Z({\rm Tube}_{\mathcal C}(\mathcal M_A)),
\]
provided the same admissibility conditions are imposed: complete positivity of
the cut-local implementing instruments, boundary-module covariance, and sequential
stability.  Thus Ising-specific data enter only in the choice of
\(\mathcal C\), \(\mathcal M_A\), and the protocol family, not in the
variational principle.
If \(\pi\) has a nontrivial fiber and the protocol family can vary the
conditional distribution inside that fiber at fixed pushforward, then
Eq.~\eqref{eq:sm-strict-functional-gap} gives a strict restricted
hypothesis-testing gap.  If no such fiber or no such protocol exists, the
theory predicts no tube-sector advantage for that measurement class.

This gives immediate non-Ising benchmarks.  In a Fibonacci anyonic chain,
\(\tau\times\tau=1+\tau\), there is no nontrivial invertible group resolution;
therefore any nontrivial tube fiber selected by an entangling-cut boundary
module would be a purely categorical readout beyond scalar overlap.  In a
Tambara-Yamagami or three-state Potts duality setting~\cite{TambaraYamagami:1998}, invertible sectors and
noninvertible duality defects coexist, so the same scalar/coarse/tube hierarchy
can be tested with endpoint Hom spaces that differ from product-KW.  These
examples are not assumed in the proof; they are the next category-independent
falsifiability tests of the variational readout principle.

In the doubled-Ising benchmark the invertible subgroup is
\(G=\mathbb Z_2\times\mathbb Z_2=\{1,a,b,c\}\), and the product
Kramers--Wannier defect \(N\) obeys
~\cite{KramersWannier:1941I,KramersWannier:1941II,KadanoffCeva:1971,Frohlich:2004ef,Aasen:2016dop}
\begin{equation}
  gN=Ng=N,
  \qquad
  N^2=1+a+b+c.
  \label{eq:sm-product-kw-fusion}
\end{equation}
The symbols \(1,a,b,c,N\) denote both topological-defect labels and the
corresponding endpoint-line operators in the displayed algebra.  Here \(g\)
runs over the invertible subgroup \(G\).
The lifted tube-sector projectors are
\begin{equation}
  P_{N_\pm}^{\rm tube}
  =
  \frac18(1+a+b+c)\pm\frac14N,
  \label{eq:sm-npm-projectors}
\end{equation}
with the three nontrivial \(G\)-character projectors
\begin{align}
  P_{+-}^{\rm tube}&=\frac14(1+a-b-c),\\
  P_{-+}^{\rm tube}&=\frac14(1-a+b-c),\\
  P_{--}^{\rm tube}&=\frac14(1-a-b+c).
\end{align}
Here \(P_{\pm\pm}^{\rm tube}\) project onto the \(G\)-character sectors with
the indicated eigenvalues under \(a\) and \(b\) (and \(c=ab\)), while
\(P_{N_\pm}^{\rm tube}\) are the two product-KW idempotents inside the
\(G\)-trivial fiber.
Thus product-KW resolution refines only the \(G\)-trivial fiber.

\subsection{Boundary-Module and SymTFT Origin of the Free Endpoint}

This subsection gives the product-KW realization of the physical endpoint
morphism principle.  The universal statement is that a noninvertible endpoint
is a boundary-module morphism before it is represented on a subsystem Hilbert
space.  In the doubled-Ising product-KW benchmark this
morphism becomes a rectangular source-target map, rather than an operator on a
single Hilbert space.
In module-category language, topological defects act on boundary conditions:
\[
  x:\ b\mapsto x\triangleright b .
\]
Equivalently, in the SymTFT description, a bulk topological line ending on a
topological boundary is a junction between two boundary conditions.  The
source boundary channel is
\[
  b_{\rm src}=e_0=(\sigma,\sigma),
\]
and the product-KW action gives
\begin{equation}
  N\triangleright b_{\rm src}
  =
  e_1\oplus e_2\oplus e_3\oplus e_4.
  \label{eq:sm-n-action-on-source}
\end{equation}
The free endpoint selects the \(G\)-invariant target channel
\begin{equation}
  b_{\rm tgt}=v_{++}
  =
  \frac12(e_1+e_2+e_3+e_4).
  \label{eq:sm-target-channel}
\end{equation}
Thus the endpoint is a junction
\begin{equation}
  \eta_{\rm free}
  \in
  {\rm Hom}_{\mathcal M}
  \!\left(
  N\triangleright b_{\rm src}, b_{\rm tgt}
  \right),
  \label{eq:sm-free-endpoint-hom}
\end{equation}
where \(\mathcal M\) is the boundary module generated by
\(\mathcal B_{\rm orb}\).  Reversing the endpoint orientation replaces
\(\eta_{\rm free}\) by its adjoint, so the convention in
Eq.~\eqref{eq:sm-free-endpoint-hom} is not part of the physical claim.  The
physical content is the pair of endpoint-bubble identities below.

Let \(\mathsf R_A\) denote the subsystem reduction, or folding, functor that represents
boundary-module morphisms on the finite subsystem Hilbert spaces.  Applying it
to the endpoint junction gives the source-target map
\[
  N_{\rm free}=\mathsf R_A(\eta_{\rm free}):
  \mathcal H_{\rm source}\longrightarrow \mathcal H_{\rm target}.
\]
Its two endpoint-bubble compositions live in different endomorphism algebras:
\begin{equation}
  N_{\rm free}^\dagger N_{\rm free}
  =
  1+a+b+c
  =
  4P_{++},
  \qquad
  N_{\rm free}N_{\rm free}^\dagger
  =
  4I_{\rm target}.
  \label{eq:sm-free-endpoint-bubbles}
\end{equation}
The first bubble is evaluated on the source module and projects onto the
\(G\)-trivial source block, with \(P_{++}=(1+a+b+c)/4\); the second bubble is
evaluated on the target endpoint module, where \(I_{\rm target}\) is the
identity.  These equations are the module/SymTFT origin of the rectangular map
used for the physical free endpoint.  In categories with higher-dimensional
endpoint Hom spaces, the corresponding endpoint state is additional boundary
data rather than a universal scalar choice.  In the present minimal product-KW
case, with \(V=N_{\rm free}/2\), the normalized endpoint satisfies
\begin{equation}
  V^\dagger V=P_{++},
  \qquad
  VV^\dagger=I_{\rm target}.
  \label{eq:sm-free-endpoint-isometry}
\end{equation}
The endpoint also intertwines the source \(G\)-trivial idempotent with the
positive product-KW tube idempotent in the target module:
\begin{equation}
  P_{N_+}^{\rm target}V=VP_{++}=V,
  \qquad
  P_{N_-}^{\rm target}V=0.
  \label{eq:sm-free-endpoint-intertwiner}
\end{equation}
This isometry maps the \(G\)-trivial source block into \(N_+\) and leaves
\(N_-\) empty.

A square same-Hilbert closure would require an additional identification
\(\mathcal H_{\rm target}\to\mathcal H_{\rm source}\).  Such an identification
is not supplied by the boundary module or by the SymTFT boundary condition.
Therefore a same-Hilbert closure may satisfy the fusion algebra, but it is only
a comparison branch and must pass the projected-spectrum test in
Sec.~\ref{sec:sm-projected-spectrum}.

\section{Endpoint-Sign Instrument and R\'enyi-Family Response}

Let \(V_+\) and \(V_-\) denote the two coherent endpoint-sign branches, with
\(V_+\) selecting the \(N_+\) tube branch and \(V_-\) selecting \(N_-\).  For
\(0\le r\le1\), the unread endpoint-sign instrument is implemented by
\begin{equation}
  K_+=\sqrt r\,\lvert+\rangle_R\otimes V_+,
  \qquad
  K_-=\sqrt{1-r}\,\lvert-\rangle_R\otimes V_-,
\end{equation}
followed by tracing out the endpoint register \(R\).  At the level of the
endpoint-resolved linear moment,
\begin{equation}
  Z_N^{(r)}
  =
  rZ_N^{(+)}+(1-r)Z_N^{(-)}
  =
  \frac12\Gsum(2r-1).
  \label{eq:sm-mixture-zn}
\end{equation}
Here \(Z_N^{(\pm)}\) are the linear product-KW moments of the pure endpoint-sign
branches and \(\Gsum\equiv Z_1+Z_a+Z_b+Z_c\).
Writing \(p_0=\Gsum/(4Z_1)\), the tube split is
\begin{equation}
  p_{N_+}^{(r)}=rp_0,
  \qquad
  p_{N_-}^{(r)}=(1-r)p_0.
  \label{eq:sm-mixture-split}
\end{equation}
Therefore the Shannon gain is
\begin{equation}
  \Delta S_{\rm tube}^{(1)}(t,r)
  =
  p_0(t)h_2(r),
  \qquad
  h_2(r)=-r\log r-(1-r)\log(1-r).
  \label{eq:sm-shannon-gain}
\end{equation}
For \(n>0\), \(n\ne1\), define
\(A_n(t)=\sum_{\chi\ne0}p_\chi(t)^n\), where \(\chi\) runs over the three
nontrivial \(G\)-character sectors.  The R\'enyi-family gain is
\begin{equation}
  \Delta S_{\rm tube}^{(n)}(t,r)
  =
  \frac{1}{1-n}
  \log
  \frac{
  A_n(t)+p_0(t)^n\left[r^n+(1-r)^n\right]
  }{
  A_n(t)+p_0(t)^n
  }.
  \label{eq:sm-renyi-gain}
\end{equation}

The same endpoint split also gives the sector-readout hypothesis test for the
endpoint protocols.  Two unread endpoint prescriptions with parameters \(r\) and
\(r'\) have the same scalar overlap and the same \(G\)-only weights, but their
product-KW tube distributions differ by
\begin{equation}
  D_{\rm TV}^{\rm tube}(r,r')
  =
  \frac12
  \sum_{\alpha=N_\pm}
  \left|p_\alpha^{(r)}-p_\alpha^{(r')}\right|
  =
  p_0|r-r'|,
  \qquad
  D_{\rm TV}^{G}=0.
  \label{eq:sm-tube-tv}
\end{equation}
For equal prior probabilities the corresponding optimal single-shot success
probability for this specified sector readout is
\begin{equation}
  P_{\rm succ}^{\rm tube}
  =
  \frac12\left(1+D_{\rm TV}^{\rm tube}\right).
  \label{eq:sm-sector-readout-success}
\end{equation}
This is a statement about the operational tube-label measurement, not a claim
about the trace distance between the full many-body states.

\subsection{MPO and Swap-Test Implementation}

The readout can be implemented without tomography of \(O_A\).  On the lattice,
the generalized symmetry defect is represented by an MPO segment ending on the
entangling-cut boundary module.  Inserting the idempotent network for
\(e_\alpha\) at the endpoint register realizes the sector projector
\(P_\alpha=\mathsf R_A(e_\alpha)\).  A replica or swap estimator then measures
the positive functional
\begin{equation}
  w_\alpha
  =
  \Tr(P_\alpha O_A),
  \qquad
  p_\alpha=\frac{w_\alpha}{\sum_\beta w_\beta}.
\end{equation}
Equivalently, one may view the procedure as an ancilla-assisted quantum
instrument: the endpoint register is coherently coupled to the defect MPO,
measured in the central-idempotent basis, and then either retained or
classically forgotten.  Summing the outcomes in each fiber implements the
ordinary \(G\)-resolved readout, while tracing the endpoint sign implements the
unread endpoint-sign channel.  Thus the tube readout is an operational
restricted measurement, not a post hoc decomposition of the Hilbert space.

\begin{figure}[t]
\centering
\includegraphics[width=0.86\textwidth]{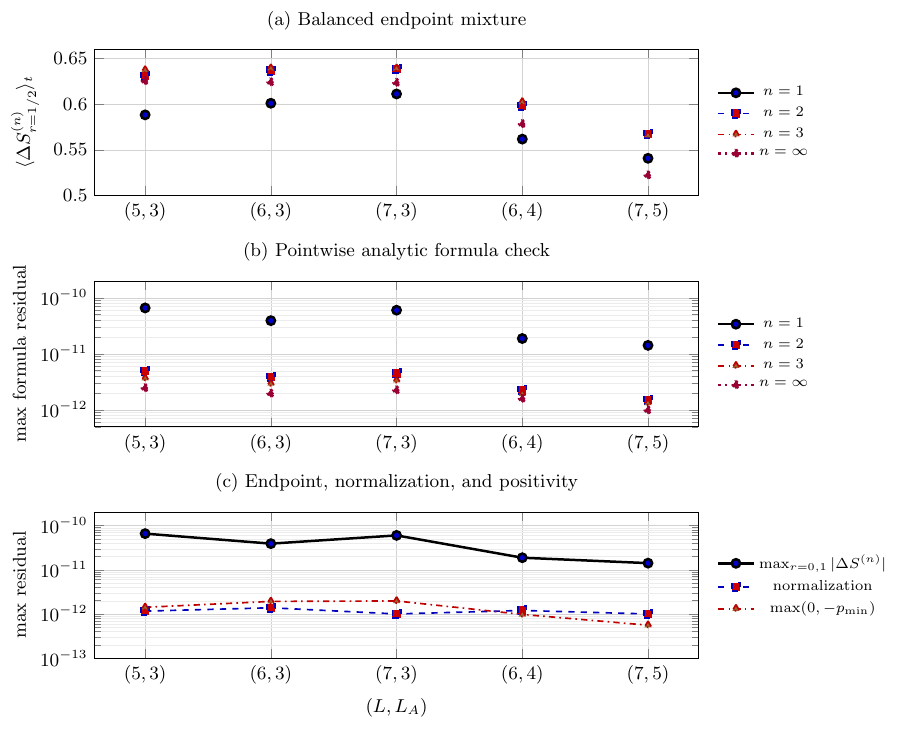}
\caption{Appendix robustness check for the R\'enyi-family endpoint-mixture
response.  The scan repeats the analytic endpoint-sign formula over
\((L,L_A)=(5,3),(6,3),(7,3),(6,4),(7,5)\), using raw \(N_+\) and \(N_-\) branch
data.}
\label{fig:sm-renyi-robustness}
\end{figure}

\subsection{Numerical Support for the Summary Diagrams}

The manuscript uses three compact summary diagrams to display the
measurement-algebra hierarchy, the conditional tube-fiber mechanism, and the
operational hypothesis test.  The figures in this subsection give the
corresponding numerical evidence.  They are deliberately kept here so that the
main text remains focused on the theorem-level logic while the raw
product-KW checks remain visible and reproducible.

\begin{figure}[t]
\centering
\includegraphics[width=0.86\textwidth]{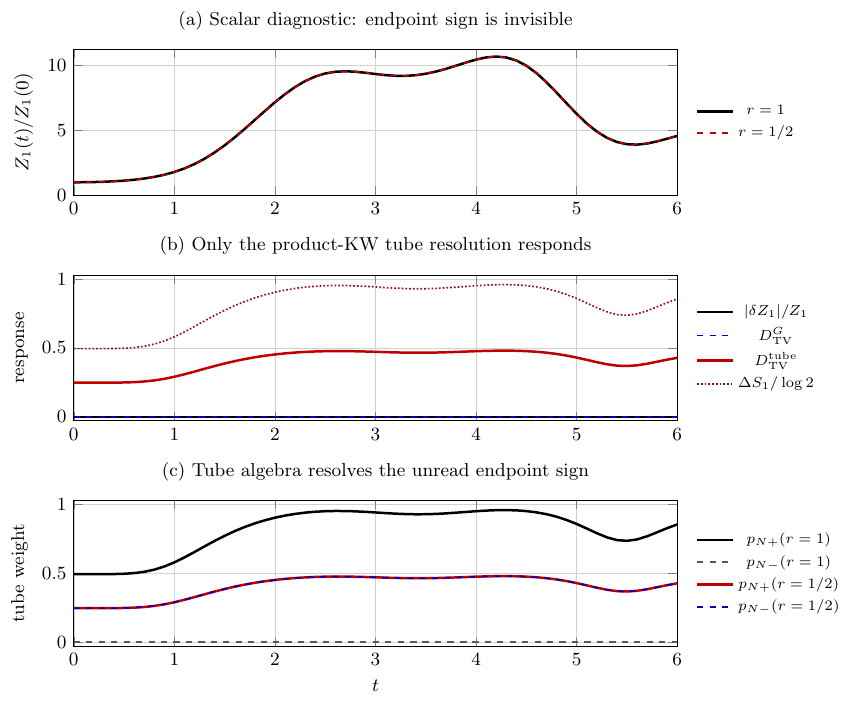}
\caption{Numerical support for the measurement-hierarchy and conditional-fiber
diagrams.  The coherent endpoint and the unread endpoint-sign instrument have
the same scalar overlap and the same \(G=\mathbb Z_2\times\mathbb Z_2\)
sector weights, but differ after resolving the \(G\)-trivial product-KW tube
fiber \(\{N_+,N_-\}\).  This is the finite-size realization of the
forgetful-map hierarchy: scalar and \(G\)-resolved readouts are blind, while
the tube readout detects the conditional fiber split.}
\label{fig:sm-discriminator-original}
\end{figure}

\begin{figure}[t]
\centering
\includegraphics[width=0.86\textwidth]{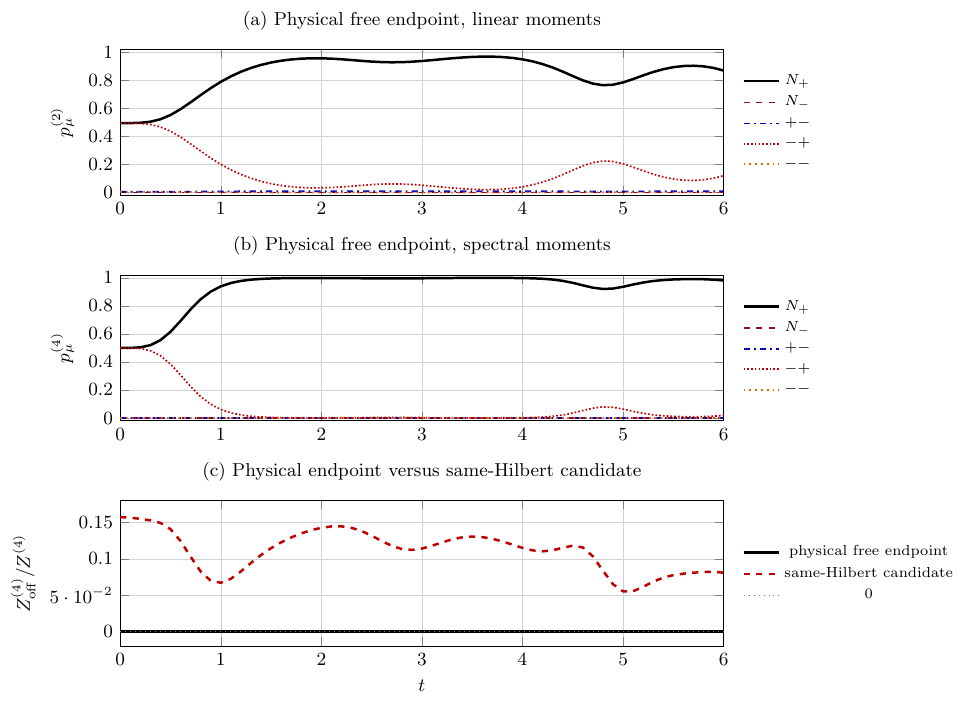}
\caption{Numerical support for the physical endpoint prescription underlying
the tube readout.  The physical free endpoint behaves as the rectangular
boundary-changing isometry: the coherent branch maps the \(G\)-trivial source
block into \(N_+\), leaves \(N_-\) empty, and removes the off-tube projected
\(m=4\) spectral component to numerical precision.  A same-Hilbert
product-KW closure is shown as a control; it can satisfy linear fusion checks
but retains off-tube spectral coherence.}
\label{fig:sm-physical-endpoint-projected-original}
\end{figure}

\begin{figure}[t]
\centering
\includegraphics[width=0.86\textwidth]{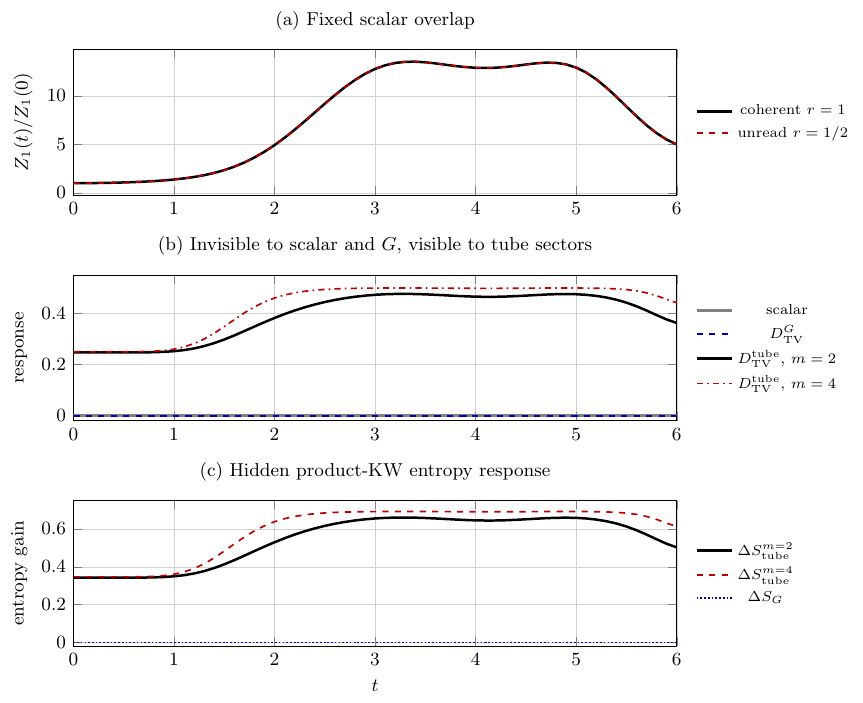}
\caption{Numerical support for the operational hypothesis-testing diagram.
For coherent versus balanced unread endpoint protocols, scalar and
\(G\)-resolved total-variation distances vanish, while the product-KW
tube-sector total variation is nonzero.  The signal is fixed by the measured
\(G\)-trivial weight \(p_0(t)\), giving the predicted success-probability gain
\(\Delta P=D_{\rm TV}^{\rm tube}/2\) for the sector-readout test.}
\label{fig:sm-conditional-discrimination-original}
\end{figure}

\section{General Fiber Derivation}
\label{sec:sm-general-fiber}

This section gives the abstract derivation of conditional tube discrimination
behind the product-KW endpoint calculation.  Let \(\mathcal T\) denote a tube-sector resolution and
\(\mathcal Q\) a coarser resolution, related by a forgetful map
\(\pi:\mathcal T\to\mathcal Q\).  For a protocol \(\lambda\), let
\(p_\alpha(\lambda)\) be the normalized tube-sector distribution.  The coarse
distribution and the conditional distribution inside each fiber are
\begin{equation}
  p_q(\lambda)
  =
  \sum_{\alpha:\pi(\alpha)=q}p_\alpha(\lambda),
  \qquad
  p_{\alpha|q}(\lambda)
  =
  \frac{p_\alpha(\lambda)}{p_q(\lambda)}
  \quad (p_q>0).
  \label{eq:sm-fiber-decomposition}
\end{equation}
Equivalently,
\begin{equation}
  p_\alpha(\lambda)
  =
  p_q(\lambda)p_{\alpha|q}(\lambda),
  \qquad
  \alpha\in\pi^{-1}(q).
\end{equation}
Thus any scalar or \(\mathcal Q\)-only diagnostic that depends only on the
coarse weights \(p_q\) cannot see changes of \(p_{\alpha|q}\) at fixed \(p_q\).

\subsection{Irreducibility relative to a coarse readout}

The previous statement can be formulated as a precise irreducibility claim.
Let
\[
  \pi_*:\mathbb R^{|\mathcal T|}\longrightarrow\mathbb R^{|\mathcal Q|}
\]
be the linear stochastic map defined by
\((\pi_*p)_q=\sum_{\alpha:\pi(\alpha)=q}p_\alpha\).  A scalar, group-resolved,
or defect-fusion diagnostic that factors through the coarse readout is a
function of \(\pi_*p\).  Therefore every deformation
\(\delta p\in\ker\pi_*\) is invisible to such a diagnostic:
\[
  \pi_*(p+\delta p)=\pi_*p,
  \qquad
  \sum_{\alpha:\pi(\alpha)=q}\delta p_\alpha=0
  \quad \forall q .
\]
If \(p+\delta p\) remains a normalized positive distribution and
\(\delta p\ne0\), the tube readout distinguishes the two distributions by
\[
  D_{\rm TV}^{\rm tube}(p,p+\delta p)
  =
  \frac12\sum_{\alpha\in\mathcal T}|\delta p_\alpha|>0,
  \qquad
  D_{\rm TV}^{\mathcal Q}=0 .
\]
This proves irreducibility relative to the chosen coarse measurement algebra:
no postprocessing of scalar or \(\mathcal Q\)-only data can reconstruct a
nonzero kernel component.  The statement does not claim that tube idempotents
are new categorical sectors; it claims that their kernel under the
tube-to-coarse pushforward can store operational distinguishability of
positive-overlap protocols.

For the product-KW fiber this kernel is one-dimensional.  At fixed
\(G\)-trivial weight \(p_0\), the deformation
\[
  \delta p_{N_+}=\epsilon p_0,\qquad
  \delta p_{N_-}=-\epsilon p_0
\]
has zero pushforward to the \(G\)-character readout but gives
\[
  D_{\rm TV}^{\rm tube}=|\epsilon|p_0 .
\]
The coherent endpoint \(r=1\) and the unread endpoint sign \(r=1/2\) correspond
to \(|\epsilon|=1/2\), giving \(D_{\rm TV}^{\rm tube}
=p_0/2\) while all scalar and \(G\)-resolved weights agree.

For the Shannon entropy,
\begin{align}
  S_{\rm tube}^{(1)}
  &=
  -\sum_{\alpha\in\mathcal T}p_\alpha\log p_\alpha \nonumber\\
  &=
  -\sum_{q\in\mathcal Q}
  \sum_{\alpha:\pi(\alpha)=q}
  p_qp_{\alpha|q}
  \log\!\left(p_qp_{\alpha|q}\right) \nonumber\\
  &=
  -\sum_q p_q\log p_q
  -
  \sum_q p_q
  \sum_{\alpha:\pi(\alpha)=q}
  p_{\alpha|q}\log p_{\alpha|q}.
\end{align}
Therefore
\begin{equation}
  S_{\rm tube}^{(1)}
  =
  S_{\mathcal Q}^{(1)}
  +
  \sum_{q\in\mathcal Q}p_q\,H(p_{\alpha|q}),
  \label{eq:sm-shannon-chain-rule}
\end{equation}
Here \(S_{\mathcal Q}^{(1)}=-\sum_qp_q\log p_q\), and
\(H(p_{\alpha|q})=-\sum_{\alpha\in\pi^{-1}(q)}
p_{\alpha|q}\log p_{\alpha|q}\) is the Shannon entropy inside the fiber over
\(q\).  This is the ordinary entropy chain rule applied to the forgetful
tube-to-coarse map.

The R\'enyi family is similarly spectral.  For \(n>0\), \(n\ne1\),
\begin{align}
  S_{\rm tube}^{(n)}
  &=
  \frac{1}{1-n}
  \log\sum_{\alpha\in\mathcal T}p_\alpha^n \nonumber\\
  &=
  \frac{1}{1-n}
  \log
  \sum_{q\in\mathcal Q}
  p_q^n
  R_q^{(n)},
  \qquad
  R_q^{(n)}
  =
  \sum_{\alpha:\pi(\alpha)=q}p_{\alpha|q}^{\,n}.
  \label{eq:sm-renyi-fiber-form}
\end{align}
Since
\begin{equation}
  S_{\mathcal Q}^{(n)}
  =
  \frac{1}{1-n}
  \log\sum_q p_q^n,
\end{equation}
the tube readout refinement over the coarse resolution is
\begin{equation}
  \Delta S^{(n)}
  =
  \frac{1}{1-n}
  \log
  \frac{\sum_q p_q^nR_q^{(n)}}{\sum_q p_q^n}.
  \label{eq:sm-general-renyi-refinement}
\end{equation}
Thus the refinement is not a peculiarity of Shannon entropy; it is a change
in the projected tube spectrum.

Now compare two protocols \(\lambda\) and \(\lambda'\) with the same coarse
distribution \(p_q(\lambda)=p_q(\lambda')\).  The total variation distance
visible to the coarse readout algebra
\[
  \mathcal A_{\mathcal Q}
  =
  {\rm span}\!\left\{
  \sum_{\alpha:\pi(\alpha)=q}P_\alpha
  \right\}_{q\in\mathcal Q}
\]
is zero, but the tube readout gives
\begin{align}
  D_{\rm TV}^{\rm tube}
  &=
  \frac12
  \sum_{\alpha\in\mathcal T}
  \left|
  p_\alpha(\lambda)-p_\alpha(\lambda')
  \right| \nonumber\\
  &=
  \frac12
  \sum_{q\in\mathcal Q}
  p_q
  \sum_{\alpha:\pi(\alpha)=q}
  \left|
  p_{\alpha|q}(\lambda)
  -
  p_{\alpha|q}(\lambda')
  \right|.
  \label{eq:sm-general-tv}
\end{align}
For equal prior probabilities, the optimal success probability for the
specified sector readout is
\begin{equation}
  P_{\rm succ}^{\rm tube}
  =
  \frac12\left(1+D_{\rm TV}^{\rm tube}\right),
  \qquad
  P_{\rm succ}^{\mathcal Q}
  =
  \frac12 .
  \label{eq:sm-general-success}
\end{equation}
The second equality uses the assumed equality of the coarse distributions
\(p_q\).  The gain \(D_{\rm TV}^{\rm tube}/2\) is therefore a classical
hypothesis-testing advantage for the chosen tube-label measurement, not a
new state metric.  Equivalently, for the target family
\(\mathcal P_\pi\) of pairs with identical coarse pushforward and different
conditional tube distributions,
\begin{equation}
  \mathcal F_{\mathcal P_\pi}(\mathcal A_{\mathcal Q})
  =
  \frac12
  <
  \mathcal F_{\mathcal P_\pi}(\mathcal A_{\rm phys}) .
  \label{eq:sm-strict-functional-gap}
\end{equation}
This strict inequality is the operational content of the no-go statement:
every admissible \(\mathcal Q\)-only readout has already discarded the
conditional tube record.  It is not a claim about the trace distance between
the full many-body states.

For the product-KW endpoint, the only nontrivial fiber is the \(G\)-trivial
one:
\begin{equation}
  \pi^{-1}(0)=\{N_+,N_-\}.
\end{equation}
Here \(0\) denotes the trivial \(G\)-character sector.
With endpoint-sign mixture parameter \(r\),
\begin{equation}
  p_{N_+}^{(r)}=rp_0,
  \qquad
  p_{N_-}^{(r)}=(1-r)p_0,
\end{equation}
so the conditional distribution in this fiber is simply
\((r,1-r)\).  All other fibers are singletons.  Equations
\eqref{eq:sm-shannon-chain-rule} and \eqref{eq:sm-general-renyi-refinement}
therefore reduce to
\begin{equation}
  \Delta S^{(1)}
  =
  p_0h_2(r),
  \qquad
  \Delta S^{(n)}
  =
  \frac{1}{1-n}
  \log
  \frac{
  A_n+p_0^n\left[r^n+(1-r)^n\right]
  }{
  A_n+p_0^n
  },
  \label{eq:sm-product-kw-general-reduction}
\end{equation}
where \(A_n=\sum_{\chi\ne0}p_\chi^n\).  Similarly,
\begin{equation}
  D_{\rm TV}^{\rm tube}(r,r')
  =
  p_0|r-r'|,
  \qquad
  D_{\rm TV}^{G}=0.
  \label{eq:sm-product-kw-tv-reduction}
\end{equation}
This is the abstract reason why the coherent endpoint and the unread
endpoint-sign instrument can have identical scalar and \(G\)-only diagnostics
while remaining distinguishable by the product-KW tube readout.

\section{Projected-Spectrum and Finite-Size Robustness}
\label{sec:sm-projected-spectrum}

The projected-spectrum test checks more than the linear tube weights.  It asks
whether the endpoint prescription block diagonalizes the positive overlap
spectrum:
\begin{equation}
  \frac{Z_{\rm off}^{(4)}}{Z^{(4)}}
  =
  1-\frac{\sum_\alpha Z_\alpha^{(4)}}{\Tr O_A^2}.
  \label{eq:sm-off-tube}
\end{equation}
Here \(Z^{(4)}=\Tr O_A^2\), \(Z_\alpha^{(4)}\) is the fourth moment retained by
the projected tube sector \(\alpha\), and \(Z_{\rm off}^{(4)}\) is the
corresponding off-tube spectral fraction.
For the physical endpoint the expected identities are
\begin{equation}
  Z_{N_+}^{(2k)}
  =
  \Tr\!\left[(P_{++}O_{\rm source}P_{++})^k\right],
  \qquad
  Z_{N_-}^{(2k)}=0,
  \qquad k=1,2.
  \label{eq:sm-physical-projected}
\end{equation}
The finite-size and endpoint-size scan covers
\((L,L_A)=(5,3),(6,3),(7,3),(6,4),(7,5)\).  The largest observed residuals are
\begin{equation}
\begin{aligned}
  \max |Z_{N_+}^{(2)}-Z_{P_{++}}^{(2)}| &= 3.33\times10^{-16},\\
  \max |Z_{N_+}^{(4)}-Z_{P_{++}}^{(4)}| &= 2.78\times10^{-16},\\
  \max Z_{\rm off}^{(4)}/Z^{(4)} &= 1.72\times10^{-15}.
\end{aligned}
\end{equation}
These residuals are the numerical falsifiability criteria for the
free-endpoint prescription.  Keeping the boundary-module identity fixed, a
stable nonzero \(N_-\) projected weight or a nonzero thermodynamic trend in
\(Z_{\rm off}^{(4)}/Z^{(4)}\) would invalidate the rectangular-isometry
endpoint.  The comparison same-Hilbert closures fail
precisely this stronger projected-spectrum test.

\begin{figure}[t]
\centering
\includegraphics[width=0.86\textwidth]{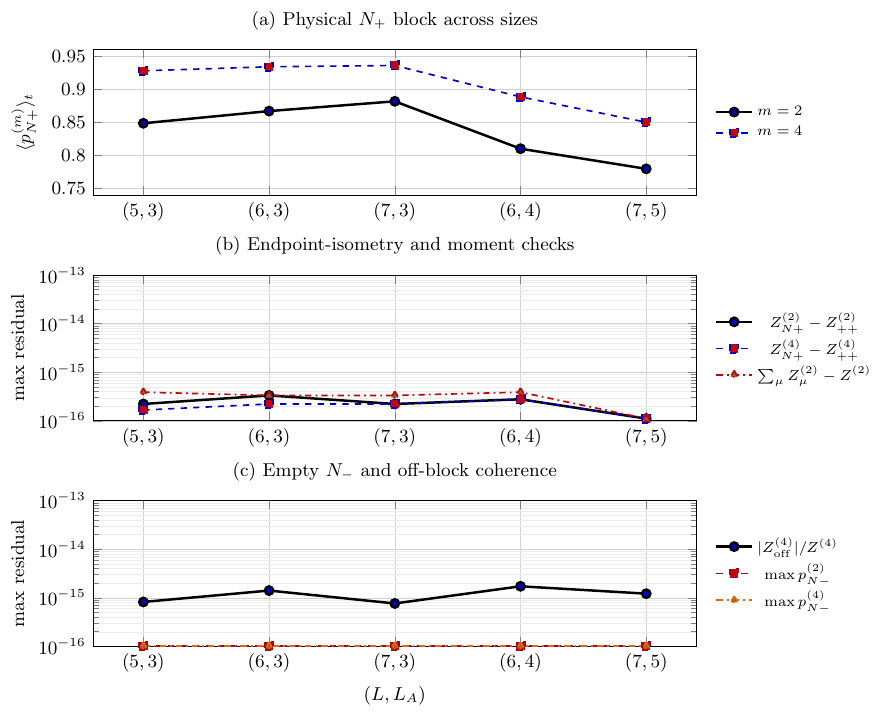}
\caption{Appendix finite-size and endpoint-size robustness of the physical
free-endpoint projected spectrum.  The physical endpoint keeps \(N_-\) empty
and removes the \(m=4\) off-tube fraction to numerical precision across the
tested cases.}
\label{fig:sm-physical-endpoint-robustness}
\end{figure}

The full \((L,L_A)=(7,5)\), 61-point production curve was generated with a
Schmidt-factor implementation.  Writing \(O_A(t)=B_tB_t^\dagger\), the projected
moments are computed from \(B_t^\dagger P_\alpha B_t\), reducing the spectral
step from the \(2^{10}\)-dimensional subsystem space to rank
\(2^{2(L-L_A)}=16\) projected Gram matrices.

\begin{figure}[t]
\centering
\includegraphics[width=0.86\textwidth]{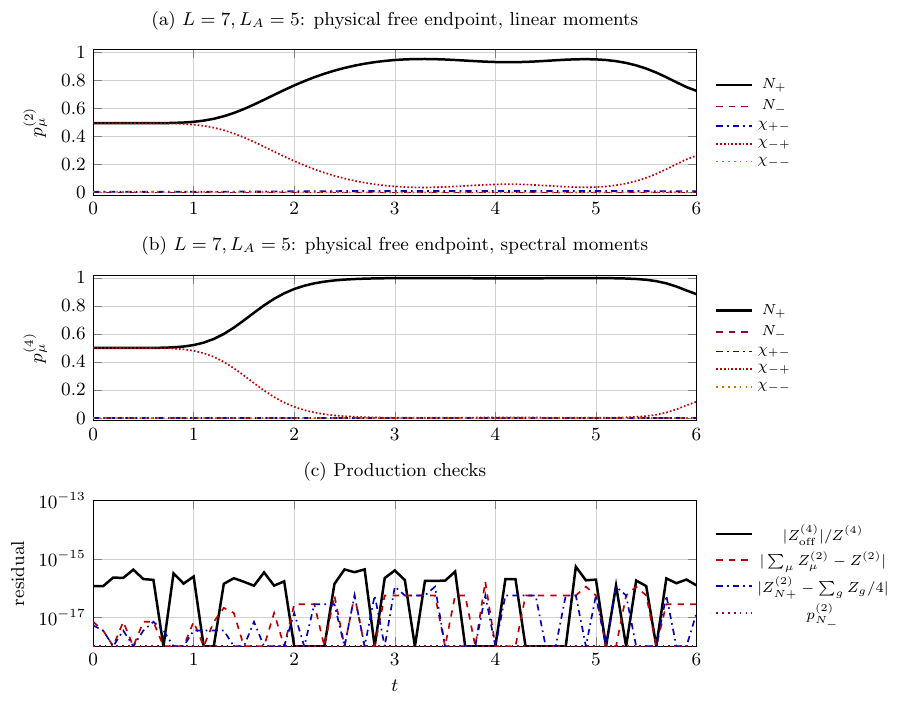}
\caption{Appendix 61-point production curve for the
\((L,L_A)=(7,5)\) physical free endpoint.  The optimized implementation is
cross-checked against a dense reference calculation and preserves the
projected-spectrum identities to roundoff.}
\label{fig:sm-la5-production}
\end{figure}

\section{Control Diagnostics}

This section collects checks that rule out simpler or misleading
interpretations.

\subsection{Projector-Hierarchy Control: Naive and Same-Hilbert Product-KW Closures}

The \(G\)-only resolution is a valid coarse graining but is endpoint blind.  A
square same-Hilbert product-KW closure goes further and can satisfy the fusion
algebra, but it produces an unphysical \(N_-\) branch and finite \(m=4\)
off-tube coherence.  The boundary-tube/free-endpoint construction repairs both
issues.  This is the control hierarchy for the endpoint construction: the same
linear fusion checks that make the square closure tempting are not sufficient
to validate a physical endpoint prescription.

\begin{figure}[t]
\centering
\includegraphics[width=0.86\textwidth]{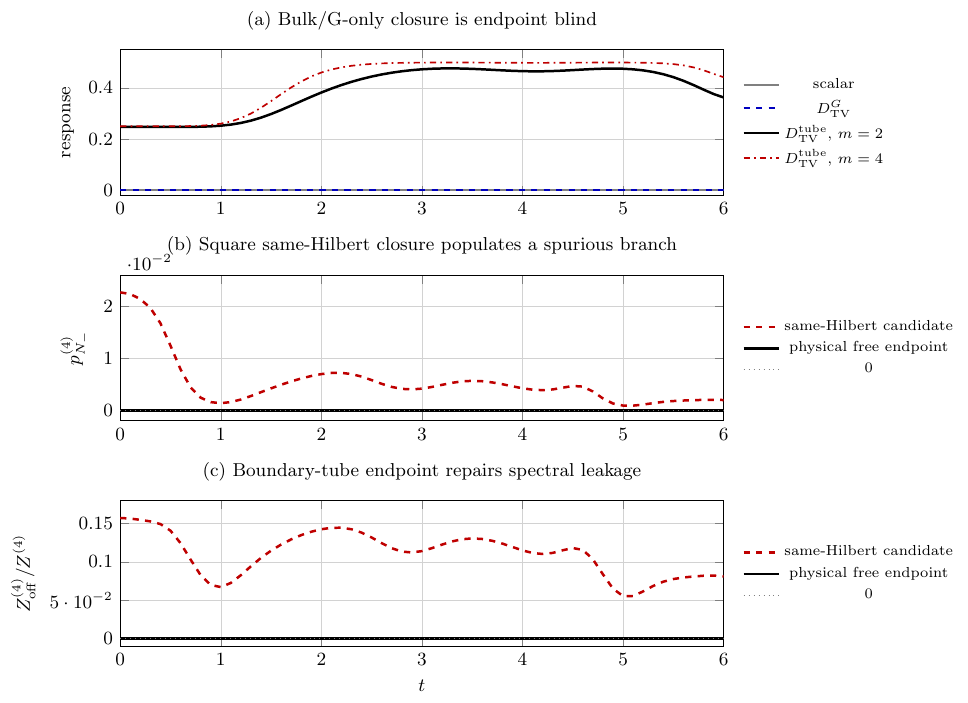}
\caption{Control hierarchy for naive, same-Hilbert, and physical free-endpoint
resolutions.}
\label{fig:sm-naive-repair}
\end{figure}

\subsection{Positive-Overlap Complement Anomaly}

The positive overlap operator is not a transition matrix.  For different pure
states, \(\Tr(\rho_A^\psi\rho_A^\phi)\) need not equal
\(\Tr(\rho_{\bar A}^\psi\rho_{\bar A}^\phi)\).  Complement mismatch is
therefore a property of the positive-overlap diagnostic, not a failure of the
tube projectors.  The tested cases remain internally normalized and block
diagonal to roundoff.

The robustness scan covers eight finite \((L,L_A,\mathrm{copy})\) cases and
records scalar, sector, and off-block residuals.
The provenance table below lists the script and summary output used for this
check.

\begin{figure}[t]
\centering
\includegraphics[width=0.86\textwidth]{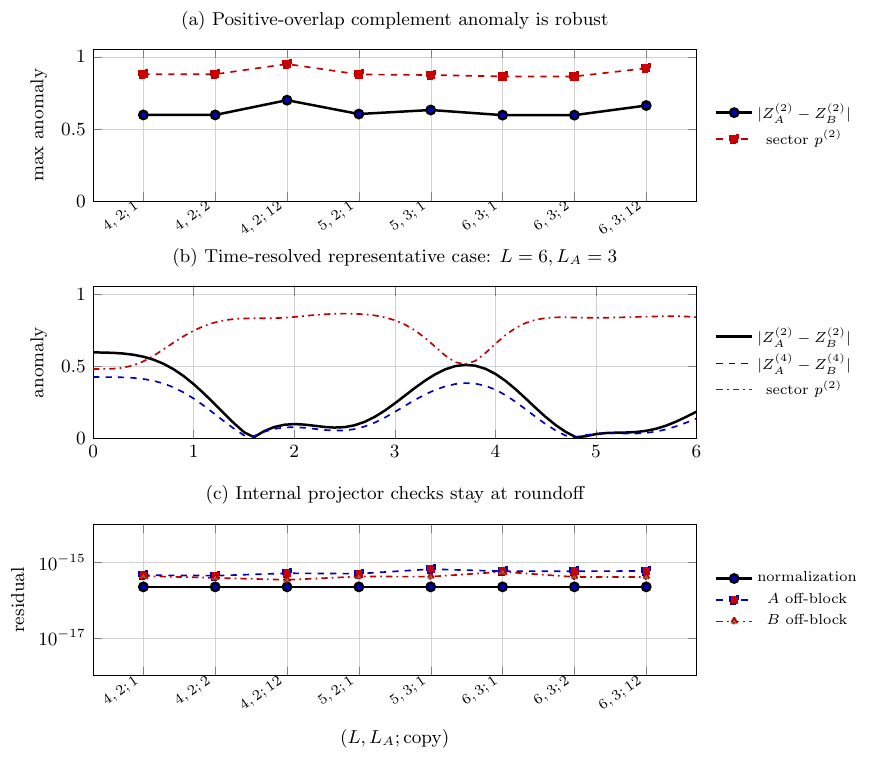}
\caption{Complement anomaly of the positive-overlap diagnostic.  Internal
projector checks remain at roundoff even though scalar and sector-resolved
\(A/\bar A\) complement identities fail.}
\label{fig:sm-complement-anomaly}
\end{figure}

\section{Local-Quench CFT Support and Data Provenance}

All doubled-Ising and CFT computations in this section serve as benchmark
realizations and consistency checks of the variational readout principle; they
are not assumptions in its derivation.

\subsection{CFT Formula for the Product-KW Source Weight}

The local-quench plateau has two logically separate
pieces.  The continuum CFT piece is a boundary-CFT block decomposition of the
replica correlator~\cite{Calabrese:2004eu,Calabrese:2005in,Calabrese:2007mtj,Nozaki:2014hna,He:2014mwa,Guo:2018lzz}.
The microscopic piece is the map from the finite lattice
time \(t\) to the continuum cross-ratios \(x(t),\bar x(t)\).

This section supplies the support for the claim that the plateau is
controlled by conformal-block powers rather than by an independent numerical
fit.

For a local primary \(\Phi\), let \(W_a(x,\bar x;\Phi)\) denote the positive
weight of boundary channel \(a\) in the positive-overlap replica geometry.
In the continuum rational conformal field theory (RCFT) block resolution, and
when the projected readout is
diagonal in the relevant boundary-channel basis, a categorical projector
selecting a set \(\mathcal I\) of boundary channels gives normalized projected
even moments of the channel-power form
\begin{equation}
  p_{\mathcal I}^{(2k)}(t)
  =
  \frac{\sum_{a\in\mathcal I}W_a(x(t),\bar x(t);\Phi)^k}
       {\sum_b W_b(x(t),\bar x(t);\Phi)^k}.
  \label{eq:sm-cft-channel-power-general}
\end{equation}
The denominator sums over all boundary channels \(b\) in the same replica
geometry.  The finite-size projected-spectrum tests reported below check that
off-channel leakage is negligible for the microscopic implementation used in
the numerical comparison.
For \(k=1\) this is the ordinary block-weight ratio; for \(k>1\) it is the
same channel decomposition applied to the projected positive-overlap spectrum.
Thus the \(m=4\) plateau sharpening is a spectral statement: the dominant
boundary block is squared relative to the subdominant block.

For the Ising spin local quench at the free entangling boundary,
\(\sigma\times\sigma=1+\psi\).  Up to a common prefactor, the two spin
conformal blocks are
\begin{equation}
  \mathcal F_{1,\psi}^{\sigma}(x)
  \propto
  \sqrt{1\pm\sqrt{1-x}},
  \label{eq:sm-ising-spin-blocks}
\end{equation}
so the real block coordinate gives
\begin{equation}
  p_+^{(2)}(x)=\frac{1+\sqrt{1-x}}2.
  \label{eq:sm-pplus-m2}
\end{equation}
The identity-channel member of Eq.~\eqref{eq:sm-cft-channel-power-general}
therefore gives
\begin{equation}
  p_+^{(2k)}(x)
  =
  \frac{W_1(x)^k}{W_1(x)^k+W_\psi(x)^k}.
  \label{eq:sm-pplus-family}
\end{equation}
For \(m=4\),
\begin{equation}
  p_+^{(4)}
  =
  \frac{(p_+^{(2)})^2}{(p_+^{(2)})^2+(1-p_+^{(2)})^2}.
  \label{eq:sm-pplus-m4}
\end{equation}

In the doubled-Ising production data the excitation is in copy 1, while copy 2
is a spectator.  Let \(p_{+,j}^{(2k)}\) denote the positive, or identity-block,
weight in Ising copy \(j=1,2\).  The \(G\)-trivial source weight factorizes as
\begin{equation}
  p_0^{(2k)}(t)
  =
  p_{+,1}^{(2k)}(t)\,p_{+,2}^{(2k)}(t)
  \equiv
  \kappa_{2k}(t)\,
  \mathcal F_{\rm CFT}^{(2k)}(x(t),\bar x(t);\sigma),
  \label{eq:sm-p0-cft-factorized}
\end{equation}
where \(\kappa_{2k}=p_{+,2}^{(2k)}\) is the measured spectator-copy factor.
Equivalently, in the columns of the comma-separated values (CSV) files used for
the main figure,
\(\kappa_{2}=p_{N_+}^{(2)}+p_{\chi_{-+}}^{(2)}\) and similarly for
\(m=4\), where \(\chi_{-+}\) is the \(G\)-character sector with eigenvalues
\((- ,+)\) under the two generators \(a\) and \(b\).  In the continuum
spectator limit \(\kappa_{2k}\to1\), and
\(p_0^{(2k)}\) is precisely the CFT function
\(\mathcal F_{\rm CFT}^{(2k)}(x,\bar x;\sigma)\).

The full \((L,L_A)=(7,5)\), 61-point curve gives a maximum block-power error
of \(8.3\times10^{-7}\).

\subsection{Microscopic Two-Front Positive-Overlap Geometry}

Equation~\eqref{eq:sm-pplus-m2} can be inverted directly.  After dividing out
the spectator Ising copy, the measured \(m=2\) product-KW endpoint weight gives
the effective CFT block coordinate
\begin{equation}
  x_{\rm eff}(t)=1-\left(2p_+^{(2)}(t)-1\right)^2 .
  \label{eq:sm-xeff}
\end{equation}
This is not an additional fit parameter.  It is the cross-ratio coordinate
that the CFT formula would assign to the observed \(m=2\) boundary-channel
weight.  The microscopic test is therefore whether the independently expected
local-quench geometry produces the same \(x_{\rm eff}(t)\).

For the finite open-chain interval a single half-line cross-ratio is not
adequate.  The positive-overlap geometry has two relevant cut-crossing events:
the direct right-moving front reaches the entangling cut, and the
left-moving/reflected front reaches it after reflection from the physical
boundary.  A minimal branch-selected two-front coordinate is
\begin{equation}
  x_{\rm micro}(t)
  =
  1-H_\epsilon(t-t_1)H_\epsilon(t_2-t),
  \qquad
  H_\epsilon(y)
  =
  \frac12\left(1+\frac{y}{\sqrt{y^2+\epsilon^2}}\right),
  \label{eq:sm-xmicro}
\end{equation}
with
\begin{equation}
  t_1=\frac{a-x_0}{v_R},
  \qquad
  t_2=\frac{a+x_0}{v_L}.
\end{equation}
Here \(a\) is the entangling-cut position, \(x_0\) is the local excitation
position, \(v_R\) and \(v_L\) are the right- and left-moving front velocities,
and \(\epsilon\) is the smoothing width in \(H_\epsilon\).  The coordinate is
close to \(x=1\) before the first crossing, falls
toward the identity-block regime \(x\simeq0\) between the two crossings, and
returns after the reflected crossing.  This is precisely the spacetime pattern
needed for Eq.~\eqref{eq:sm-pplus-family}: the identity block dominates in the
post-front window, and the \(k=2\) channel power drives \(p_0^{(4)}\) much
closer to one than \(p_0^{(2)}\).

The current fit gives \(v_R=1.4301729843\), \(v_L=1.095311913\), and
\(\epsilon=1.12396396396\).  The \(L_A\ge4\) subset gives the intended
continuum trend, while \(L_A=3\) shows visible finite-size and image effects.
The all-case root-mean-square (RMS) residual is \(0.2630\), while the
\(L_A\ge4\) subset improves to \(0.1015\).  For this reason the microscopic
\(x(t)\) collapse is used as a support check on the conformal-block mechanism,
not as the leading universality claim.

\begin{figure}[t]
\centering
\includegraphics[width=0.86\textwidth]{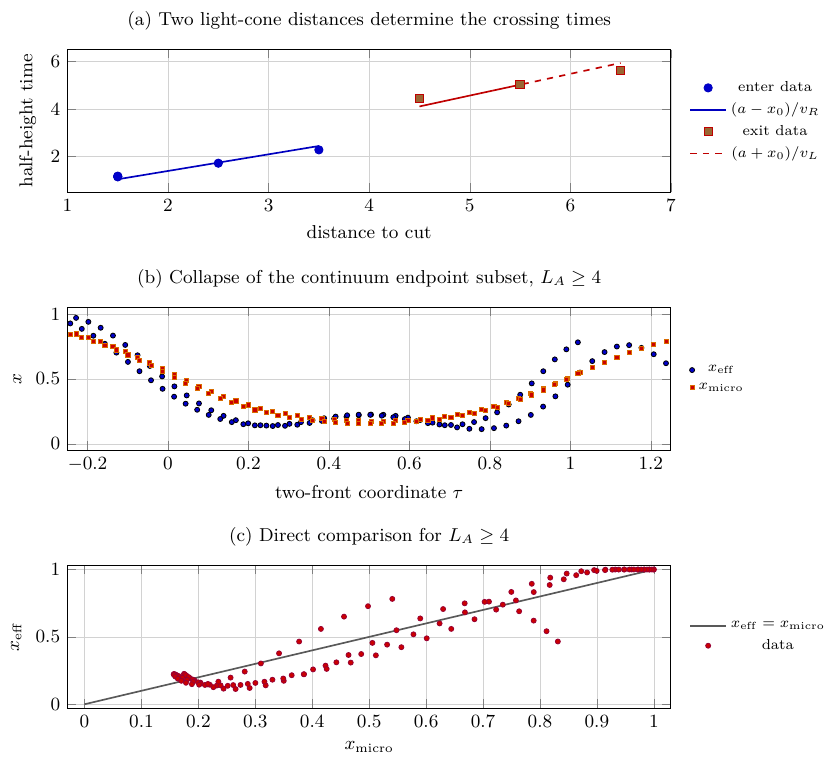}
\caption{Appendix microscopic positive-overlap geometry.  The measured
\(m=2\) endpoint weight is converted to the CFT block coordinate
\(x_{\rm eff}=1-(2p_+^{(2)}-1)^2\).  The two-front coordinate captures the
direct and reflected entangling-cut crossings that drive this coordinate into
the identity-block regime.  The \(L_A\ge4\) data show the intended continuum
trend; the \(L_A=3\) cases are retained as finite-size stress tests.}
\label{fig:sm-microscopic-x}
\end{figure}

\clearpage
\bibliographystyle{apsrev4-2}
\bibliography{references}

\end{document}